\documentclass[11pt,a4paper]{article}

\usepackage{cite}
\usepackage{graphicx}
\graphicspath{{figures/}}
\usepackage[english]{babel}
\usepackage{latexsym}
\usepackage{url}
\usepackage{amsmath,amsthm,amssymb}
\usepackage{dsfont}
\usepackage{ifthen}
\usepackage{fancybox}
\usepackage{microtype,stmaryrd}
\usepackage{algo}

\usepackage[charter]{mathdesign}

\usepackage[margin=2.55cm]{geometry}

\usepackage{color}
\definecolor{blu3}{rgb}{.1,.0,.4}
\usepackage{hyperref}
\hypersetup{colorlinks=true, linkcolor=blu3, urlcolor=blu3, citecolor=blu3, pdfpagemode=UseNone, pdfstartview=} 

\newtheorem{theorem}{Theorem}
\newtheorem{corollary}[theorem]{Corollary}
\newtheorem{lemma}[theorem]{Lemma}

\newtheorem{proposition}[theorem]{Proposition}

\newcommand{\RR}{\ensuremath{\mathbb R}}    
\newcommand{\ZZ}{\ensuremath{\mathbb Z}}    
\newcommand{\NN}{\ensuremath{\mathbb N}}    
\newcommand{\UU}{\ensuremath{\mathcal{U}}}  
\newcommand{\WW}{\ensuremath{\mathcal{W}}}  
\renewcommand{\AA}{\ensuremath{\mathcal{A}}}  
\newcommand{\GG}{\ensuremath{\mathbb{G}}}   
\newcommand{\CC}{\ensuremath{\mathbb{S}}}    
\DeclareMathOperator{\polylog}{polylog}

\DeclareMathOperator{\diam}{diam}
\DeclareMathOperator{\rank}{rank}

\def\DEF#1{\textbf{\emph{#1}}}

\begin{document}
\setcounter{page}{0} 

\title{Maximum Matchings in Geometric Intersection Graphs\thanks{A 
preliminary version of this work appeared as
\'E.~Bonnet, S.~Cabello, and
W.~Mulzer. \emph{Maximum Matchings in Geometric Intersection Graphs},
Proc. 37th STACS,2020, pp.~31:1–31:17.}}

\author{{\'E}douard Bonnet\thanks{Univ Lyon, CNRS, ENS de Lyon, 
		Universit\'e Claude Bernard Lyon 1, LIP UMR5668, France. 
		Email address: \texttt{edouard.bonnet@ens-lyon.fr}}
\and
		Sergio Cabello\thanks{Department of Mathematics, IMFM, and
                Department of Mathematics, FMF, University of Ljubljana, 
		Slovenia.
                Supported by the Slovenian Research Agency (P1-0297, 
		J1-9109, J1-8130, J1-8155, J1-1693, J1-2452, N1-0218). 
		Email address: \texttt{sergio.cabello@fmf.uni-lj.si}}
\and
	Wolfgang Mulzer\thanks{Institut f\"ur Informatik,
		Freie Universit\"at Berlin, Germany.
		Supported in part by ERC StG 757609 and GIF grant 1367.
                Email address: \texttt{mulzer@inf.fu-berlin.de}}}

\maketitle

\thispagestyle{empty}

\begin{abstract}
	Let $G$ be an intersection graph of $n$ 
	geometric objects in the plane. We show 
	that a maximum matching in $G$ can be 
	found in $O(\rho^{3\omega/2}n^{\omega/2})$ 
	time with high probability, where $\rho$ 
	is the density of the geometric objects 
	and $\omega>2$ is a constant such that 
	$n \times n$ matrices can be multiplied 
	in $O(n^\omega)$ time.

        The same result 
	holds for any subgraph of $G$, as long as 
	a geometric representation is at hand.
	For this, we combine algebraic methods, 
	namely computing the rank of a matrix 
	via Gaussian elimination, with the fact 
	that geometric intersection graphs have
	small separators.
	
	We also show that in many interesting
	cases, the maximum matching problem
	in a general geometric intersection graph can be 
	reduced to the case of bounded density.
	In particular, a maximum matching
	in the intersection graph of any 
	family of translates of a convex object in the
	plane can be found in $O(n^{\omega/2})$ time with
	high probability, and
	a maximum matching in the intersection graph
	of a family of planar disks with radii in 
	$[1, \Psi]$ can be found in
	$O(\Psi^6\log^{11} n + \Psi^{12 \omega} n^{\omega/2})$
	time with high probability.

\end{abstract}

\newpage
\section{Introduction}

Let $\UU$ be a family of (connected and compact) 
objects in $\RR^2$. The \DEF{intersection graph} 
$G_\UU$ of $\UU$ is the undirected graph with 
vertex set $\UU$ and edge set 
\[
	E(G_\UU) ~=~ \{ UV \mid U,V \in \UU, \, U \cap V \neq \emptyset \}.
\]
If the objects in $\UU$ are partitioned into 
two sets, one can also define the 
\emph{bipartite} intersection graph, a 
subgraph of $G_\UU$, in the obvious way.
Consider the particular case when $\UU$ is a 
set of disks.  Then, we call $G_\UU$ 
a \DEF{disk graph}, and if all disks
in $\UU$ have the same radius, a 
\DEF{unit-disk graph}.
Unit disk graphs are often used to model ad-hoc 
wireless communication networks and sensor 
networks~\cite{GG11,zg-wsn-04,HS95}. Disks of 
varying sizes and other shapes become relevant 
when different sensors cover different areas.
Moreover, general disk graphs 
serve as a tool to approach other problems, like the
barrier resilience problem~\cite{kumar2007barrier}.

We consider a classic optimization problem,
\emph{maximum matching}, in the setting of geometric 
intersection graphs, and introduce two new techniques, 
each interesting in its own. 
First, we provide an efficient algorithm to
compute a maximum matching in any subgraph
of the intersection graph of geometric objects 
of low density.
Second, we provide a sparsification technique to
reduce the maximum matching problem
in a geometric intersection graph to the
case of low density. The sparsification works for
convex shapes of similar sizes for which certain
range searching operations can be done efficiently.

In this paper, we use $\omega$ to denote 
a constant such that $\omega > 2$
and any two $n \times n$ matrices can be 
multiplied in time $O(n^\omega)$.\footnote{In the literature,
it is more common to assume $\omega \geq 2$. We adopt the
stronger assumption $\omega > 2$
because it simplifies the bounds. 
If $\omega=2$ is allowed, 
then the running times that we state have 
additional logarithmic factors.}

\paragraph*{Maximum matching in intersection graphs of geometric 
objects of low density.}
We first introduce some geometric concepts.
The diameter of a set $X \subset \RR^2$, denoted by $\diam(X)$,
is the supremum of the distances between any two points of $X$. 
The \DEF{density} $\rho(\UU)$ of a family $\UU$ of objects is 
\begin{equation}\label{equ:density}
	\rho(\UU) ~=~ 
	\max_{X \subseteq \RR^2} \big|\{U \in \UU \mid \diam(U) \ge \diam(X),
	\, U \cap X \neq \emptyset\}\big|.
\end{equation}
One can also define the density by considering for $X$ 
only disks. Since an object of diameter $d$ can be 
covered by $O(1)$ disks of diameter $d$, 
this changes the resulting parameter by only a constant.
(See, for example, the book by
de Berg~et~al.~\cite[Section 12.5]{BergCKO08} 
for such a definition.)
The \DEF{depth} (ply) of $\UU$ is 
the largest number of objects that cover a single point: 
\[
\max_{p \in \RR^2} \big|\{ U\in \UU\mid p\in U\}\big|.
\]
For disk graphs and square graphs, the depth and the density
are linearly related; see for example 
Har-Peled and Quanrud~\cite[Lemma 2.7]{Har-PeledQ17}. 
More generally, bounded depth and bounded density are 
equivalent whenever we consider homothets of a 
constant number of shapes. Density and depth are 
usually considered in the context of realistic
input models; see de Berg~et~al.~\cite{BergSVK02} 
for a general discussion.

Let $\GG_\rho$ be the family of \emph{subgraphs} 
of intersection graphs of geometric objects in 
the plane with density at most $\rho$.\footnote{Note that
by definition, any vertex-induced subgraph of a geometric intersection
graph of density at most $\rho$ is also a geometric intersection graph
of density at most $\rho$. Thus,
any graph in $\GG_\rho$ is obtained by omitting edges
from some geometric intersection graph of density at most $\rho$
with the same vertex set.
} 
Our goal is to compute a maximum matching in 
graphs of $\GG_\rho$, assuming the availability of
a geometric representation of the graph and 
a few basic geometric primitives on the geometric
objects. For this, we consider the density $\rho$ 
as an additional parameter. Naturally, the case 
$\rho = O(1)$ of \emph{bounded density} is of 
particular interest.

In a general graph $G = (V, E)$ with $n$ vertices
and $m$ edges, the best running time for computing 
a maximum matching in $G$ depends
on the ratio $m/n$. The classic algorithm of 
Micali and Vazirani~\cite{MicaliV80,Vazirani12} is
based on augmenting paths, and it finds a maximum 
matching in $O(\sqrt{n} m)$ time.
Mucha and Sankowski~\cite{MuchaS04} use algebraic 
tools to achieve running time $O(n^\omega)$.
As we shall see, for $G \in \GG_\rho$, we have
$m = O(\rho n)$, and this bound is asymptotically tight. 
Thus, for $G \in \GG_\rho$, the running times of these two 
algorithms become $O(\rho n^{3/2})$ and $O(n^\omega)$,
respectively.

In general \emph{bipartite} graphs, 
a recent algorithm by M\k{a}dry~\cite{Madry13} 
achieves running time roughly $O(m^{10/7})$.\footnote{In
a previous version of this paper, we claimed
that M\k{a}dry's algorithm also applies
to general (non-bipartite) graphs.
However, this does not seem to be correct.
As a consequence, we have updated the
statements of Corollary~\ref{cor:axis-parallel}
and Theorem~\ref{thm:matching_balls}.}
Efrat, Itai, and Katz~\cite{EfratIK01} 
show how to compute the maximum matching 
in bipartite unit disk graphs in $O(n^{3/2}\log n)$ 
time. Having bounded density does not help 
in this algorithm; it has $O(\sqrt{n})$
rounds, each of which needs $\Omega(n)$ time.
The same approach can be used for other 
geometric shapes if a certain
semi-dynamic data structure is available.
In particular, using the data structure of 
Kaplan et al.~\cite{KaplanMRSS17}
for additively-weighted nearest neighbors, 
finding a maximum matching in a bipartite 
intersection graph of disks
takes $O(n^{3/2}\polylog n)$ time. 
We are not aware of any similar results for 
non-bipartite geometric intersection graphs.

We show that a maximum matching in a graph of 
$\GG_\rho$ with $n$ vertices can be computed 
in $O(\rho^{3\omega/2} n^{\omega/2})=
O(\rho^{3.56} n^{1.19})$ time. The algorithm 
is randomized and succeeds with high probability.
It uses the algebraic approach by Mucha and 
Sankowski~\cite{MuchaS06} for planar graphs
with an extension by Yuster and Zwick~\cite{YusterZ07} 
for $H$-minor-free graphs. As noted by 
Alon and Yuster~\cite{AlonY13}, this approach 
works for \emph{hereditary}\footnote{A graph family is called
\emph{hereditary} if it is closed
under taking subgraphs.} graph families 
with bounded average degree and small 
separators. We note that the algorithm 
can be used for graphs of $\GG_\rho$, because 
we have average degree $O(\rho)$
and balanced separators of size 
$O(\sqrt{\rho n})$~\cite{Har-PeledQ17,SmithWo98}.
However, finding the actual dependency on 
$\rho$ is difficult because it plays a 
role in the average degree, in the size 
of the separators, and because the algorithm 
has a complex structure with several 
subroutines that must be distilled.

There are several noteworthy features in 
our approach. For one, we solve a 
geometric problem using linear algebra,
namely Gaussian elimination.
The use of geometry is limited to finding
separators, bounding the degree, and constructing
the graph explicitly.
Note that the role of subgraphs in the definition
of $\GG_\rho$ is a key feature in our algorithm.
On the one hand, we need a hereditary family of graphs,
as needed to apply the algorithm.
On the other hand, it brings more generality; for 
example, it includes the case of bipartite graphs 
defined by colored geometric objects.

Compared to the work of Efrat, Itai, and 
Katz~\cite{EfratIK01}, our algorithm is for 
arbitrary subgraphs of geometric intersection 
graphs, not only bipartite ones; it works for 
any objects, as it does not use advanced data 
structures that may depend on the shapes.
On the other hand, it needs the assumption 
of low density. Compared to previous 
algorithms for arbitrary graphs and ignoring
polylogarithmic factors, our algorithm
is faster when $\rho = o(n^{(20-7\omega)/(21\omega-20)})$.
Using the current bound $\omega < 2.373$, this means
that our new algorithm is faster for $\rho = O(n^{0.113})$.

Our matching algorithm also applies for intersection graphs 
of objects in $3$-dimensional space. However, in 
this case there is no algorithmic gain with 
the current bounds on $\omega$:
one gets a running time of $O(n^{2\omega/3})$ when $\rho=O(1)$,
which is worse than constructing the graph explicitly and using the
algorithm of Micali and Vazirani.

\paragraph*{Sparsification -- Reducing to bounded depth.}
Consider a family of convex geometric objects $\UU$ in the plane
where each object contains a square 
of side length $1$ and is contained in a square 
of side length $\Psi\ge 1$.
Our objective is to compute a maximum matching 
in the intersection graph $G_\UU$.\footnote{
Note that here we do not consider subgraphs of $G_\UU$;
we need the whole graph $G_\UU$.}
Our goal is to transform this problem to finding
a maximum matching in the intersection graph of a 
subfamily $\UU'\subset \UU$
with bounded depth. Then we can employ our result from above
for $G_{\UU'}$ or, more generally,
any algorithm for maximum matching (taking advantage
of the sparsity of the new instance). 

We describe a method that is 
fairly general and works under
comparatively mild assumptions and also in higher dimensions.
However, for an efficient implementation, we 
require that the objects under consideration
support certain range searching operations efficiently.
We discuss how this can be done for
disks of arbitrary sizes,
translates of a fixed convex shape in the plane,
axis-parallel objects in constant dimension,
and (unit) balls in constant dimension.
In all these cases, we obtain a subquadratic time
algorithm for finding a maximum matching, assuming that
$\Psi$ is small.
We mostly focus on the planar case, mentioning higher
dimensions as appropriate.

As particular results to highlight,
we show that a maximum matching
in the intersection graph of any 
family of translates of a convex object in the
plane can be found in $O(n^{\omega/2})$ time with
high probability, and
a maximum matching in the intersection graph
of a family of planar disks with radii in 
$[1, \Psi]$ can be found in
$O(\Psi^6\log^{11} n + \Psi^{12 \omega} n^{\omega/2})$
time with high probability.
See Table~\ref{tb:summary_results} for a summary of the results in this context.

\renewcommand{\arraystretch}{1.5}
\begin{table}
\centering
	\begin{tabular}{lccl}
		objects & time complexity & reference \\ \hline\hline
		disks of radius in $[1,\Psi]$ in $\RR^2$ & 
			$O(\Psi^6 n \log^{11} n + \Psi^{12\omega} n^{\omega/2})$ & 
			Theorem~\ref{thm:disks} \\ \hline
		translates of $O(1)$ convex objects in $\RR^2$ & $O(n^{\omega/2})$ & 
			Theorem~\ref{thm:translates2}\\ \hline
		axis-parallel rectangles in $\RR^2$ with edges in $[1,\Psi]$ & 
			$(1+\Psi)^{O(1)}n^{\omega/2}$ & Theorem~\ref{thm:axis-parallel}
			\\ \hline
		axis-parallel boxes in $\RR^d$ with edges in $[1,\Psi]$ & 
			$(1+\Psi)^{O(d)}n^{3/2}$ & Corollary~\ref{cor:axis-parallel}
			\\ \hline
		unit balls in $\RR^3$ or $\RR^4$ & $O(n^{3/2})$ & 
			Theorem~\ref{thm:matching_balls}\\ \hline
		unit balls in $\RR^d$, $d\ge 5$ & 
			$O(n^{\frac{2 \lceil d/2\rceil}{1+ \lceil d/2\rceil}+\varepsilon})$ & 
			Theorem~\ref{thm:matching_balls}\\ \hline \\
	\end{tabular}
	\caption{Time complexity to compute the maximum 
		matching in an intersection graph. In some cases, 
		the result is correct with high probability.}
	\label{tb:summary_results}
\end{table}

\paragraph*{Organization.}
We begin with some general definitions and
basic properties of geometric intersection graphs
(Section~\ref{sec:basics}).
Then, in the first part of the paper, 
we present the new 
algorithm for finding a maximum matching in geometric 
intersection graphs of low density (Section~\ref{sec:low_density}).
In the second part, we present our sparsification
method. This is done in two steps. First, we describe
a generic algorithm that works for general families
of shapes that have roughly the same size,
assuming that certain geometric operations can 
be performed quickly. (Section~\ref{sec:sparsification}).
Second, we explain how to implement these operations for
several specific shape families, e.g., translates
of a given convex objects and disks of bounded radius
ratio (Section~\ref{sec:sparsification_eff}).
The two parts are basically independent, where the
second part uses the result from the first part as
a black box, to state the desired running times.

\section{Basics of (geometric intersection) graphs}
\label{sec:basics}

\paragraph*{Geometric objects and Computational Model.}
Several of our algorithms work under fairly weak assumptions 
on the geometric input:
we assume that the objects in $\UU$ have
\emph{constant description complexity}. This means that the 
boundary of each object is a continuous closed curve whose 
graph is a semialgebraic set, defined by a constant number 
of polynomial equalities and inequalities of 
constant maximum degree. 
For later algorithms we restrict attention to some
particular geometric objects, like disks or squares.

To operate on $\UU$, we  
require that our computational model supports 
primitive operations that 
involve a constant number of objects of $\UU$ 
in constant time, e.g., finding 
the intersection points of two boundary curves; 
finding the intersection points between a boundary curve
and a disk or a vertical line;
testing whether a point 
lies inside, outside, or on the boundary of an object;
decomposing a boundary curve into $x$-monotone pieces, etc.
See, e.g.,~\cite{KaplanMRSS17} for a further discussion and 
justification of these assumptions.

We emphasize that in addition to the primitives
on the input objects, we do not require any special
constant-time operations. In particular, even though 
our algorithms use algebraic techniques such as fast 
matrix multiplication or Gaussian elimination, we rely only 
on algebraic operations over $\ZZ_p$, where $p=\Theta(n^4)$
is a prime.
Thus, we work only with numbers of $O(\log n)$-bits,
and assuming a standard unit-cost model for such word-sizes,
as in, e.g., the word-RAM model of computation,
we simply need to bound the number of arithmetic operations
in our algorithms.

\paragraph*{Geometric intersection graphs.}
The following well-known lemma bounds 
$|G_\UU|$ in terms of $\rho$, and the time
to construct $G_\UU$. 
We include a proof for completeness.

\begin{lemma}\label{lem:edges}
    If $\UU$ has $n$ objects and density $\rho$, then $G_\UU$ has 
    at most $(\rho - 1)n$ edges (this holds in any dimension). 
    If $\UU$ consists of objects 
    in the plane, then $G_\UU$ can be constructed in 
    $O(\rho n \log n)$ time.
\end{lemma}

\begin{proof}
    The bound on $|E(G_\UU)|$ uses a simple and well-known trick; 
    see, e.g.,~\cite[Lemma 2.6]{Har-PeledQ17}:
    we orient each edge in $G_\UU$ 
    from the object of smaller 
    diameter to the object of larger diameter.
    Then, the out-degree of each $X \in \UU$ is at
    most $\rho - 1$, since by (\ref{equ:density}), 
    there are at most $\rho$ objects $U \in \UU$
    with $\diam(U) \geq \diam(X)$ that intersect $X$, 
    with $X$ being one of them. 

    Next, we describe the construction of $G_\UU$ in the planar case; 
    see~\cite{KaplanKlMuRoSeSh19} for a similar algorithm in the context 
    of disk graphs. 
    Set $k = |E(G_\UU)|$. If $U, V \in \UU$ form
    an edge in $E(G_\UU)$,
    then either (i) their boundaries intersect; or (ii) one is 
    contained inside the other.
    To find the edges of type (i), we perform a plane 
    sweep~\cite{BentleyO79,BergCKO08}.\footnote{The original
    algorithm is described for line segments, but assuming 
    appropriate geometric primitives, it also applies to 
    continuous, $x$-monotone curves in the plane.}
    For this, we split the 
    boundary of each object into a constant number of 
    $x$-monotone pieces. We sweep a vertical line $\ell$
    across the plane, and we maintain the intersection of $\ell$
    with the pieces of the boundary curves.
    The events are the start and end points of the 
    pieces of the boundary curves, as well as their 
    pairwise intersections. There are $O(n + k)$ events.
    When we detect a boundary-boundary intersection, we add 
    the corresponding edge to the output.
    An edge can be added $O(1)$ times, so we sort the output to 
    remove duplicates. Thus, it takes $O((n + k)\log n)$ time 
    to find all edges of type (i).

    To find the edges of type (ii), we perform a second plane sweep 
    to compute the \emph{trapezoidal decomposition} of the 
    planar arrangement defined by the objects in $\UU$~\cite{BergCKO08}.
    The trapezoidal decomposition is obtained by shooting upward and 
    downward vertical rays from each $x$-extremal point on 
    a boundary curve and from each intersection between 
    two boundary curves. The rays end once they encounter a 
    boundary curve, or they go into infinity.
    This results in a subdivision of the plane into  
    $O(n + k)$ (possibly unbounded) \emph{pseudo-trapezoids}. 
    The subdivision can be computed in $O((n + k)\log n)$ 
    time. We construct the \emph{dual graph} of the trapezoidal 
    decomposition, in which the vertices are the pseudo-trapezoids, 
    and two pseudo-trapezoids are adjacent if and only if 
    their boundaries intersect in more than one point.
    We perform a DFS in the resulting dual graph, keeping track 
    of the objects in $\UU$ that contain the current pseudo-trapezoid.
    Whenever we enter an object $U \in \UU$ for the first time, 
    we generate all edges between the objects that contain the 
    current pseudo-trapezoid and $U$. This takes $O(n + k)$ time. 
    We generate all edges of type (ii), and we possibly rediscover
    some edges of type (i). Thus, we sort the output once more to 
    remove duplicates.
    The total running time is $O((n + k)\log n) = O(\rho n \log n)$. 

    We remark that using more sophisticated methods, such as 
    randomized incremental construction~\cite{Mulmuley94}, 
    it may be possible to 
    improve the running time to $O(\rho n + n\log n)$.
    However, this will not help us, because later parts 
    of the algorithm will dominate the running time.
\end{proof}

\paragraph*{Separators in geometric intersection graphs.}
The classic \emph{planar separator theorem} 
by Lipton and Tarjan~\cite{LiptonTa80,DjidjevVe97} shows that any 
planar graph can be decomposed in a balanced way by removing a
small number of vertices. Even though geometric intersection 
graphs can be far from planar, similar results are also 
available for them. These results are usually parameterized by the 
\emph{depth} of the arrangement or by the \emph{area} of the separator and 
the components~\cite{AlberFi04,EppsteinMiTe95,MillerTeThVv97}.
The following recent result provides a small separator for 
general intersection graphs of bounded density.

\begin{theorem}[Lemma 2.21 in ~\cite{Har-PeledQ17}]
\label{thm:separator}
	Let $\UU$ be a set of $n$ objects in $\RR^2$ with density $\rho$.
	In $O(n)$ expected time, we can find a circle $\CC$ such that
	$\CC$ intersects at most $c\sqrt{\rho n}$ objects of $\UU$,
	the exterior of $\CC$ contains at most $\alpha n$ elements of $\UU$,
	and the interior of $\CC$ contains at most $\alpha n$ elements of $\UU$.
	Here $0<c$ and $0<\alpha<1$ are universal constants, independent 
	of $\rho$ and $n$.
\end{theorem}

The proof of Theorem~\ref{thm:separator} goes roughly as follows: 
Pick a point in each object of $\UU$, 
compute the smallest circle $\CC'$ (or an approximation thereof) that 
contains, say, $n/20$ points, and then take a 
concentric scaled copy $\CC$ of $\CC'$, 
with scale factor uniformly at random in $[1,2]$.
With constant probability, the circle $\CC'$ 
has the desired property.  This can be checked easily 
in linear time by determining
which objects of $\UU$ are inside, outside, or intersected by $\CC$. 
In expectation, a constant number of repetitions is needed to obtain
the desired circle.

A family $\GG$ of graphs is \emph{hereditary} if for every 
$G \in \GG$, it holds that all subgraphs $H$ of $G$ are also 
in $\GG$. By definition, our family $\GG_\rho$ of 
subgraphs of geometric intersection graphs with 
density $\rho$ is hereditary.
A graph $G$ is \DEF{$\delta$-sparse} if every subgraph $H$ of $G$ 
has at most $\delta|V(H)|$ edges.
Lemma~\ref{lem:edges} implies that all graphs in $\GG_\rho$ 
are $\rho$-sparse.  

Consider a graph $G$ and a vertex $v$ of $G$.
A \DEF{vertex split} at $v$ consists of adding a pendant $2$-path 
$v v'v''$, where $v'$ and $v''$ are new vertices, and possibly 
replacing some edges $uv$ incident to $v$ by new edges $uv''$; see 
Figure~\ref{fig:split3} for a sequence of splits. 
We note that a vertex split may not replace any edges.
In this case, we are just adding a pendant path of length $2$.

Let $G'$ be a graph obtained from $G$ by a sequence of $k$ vertex 
splits. Then, the size of a maximum matching in 
$G'$ is the size of a maximum matching in $G$ plus $k$. 
Furthermore, from a maximum matching in $G'$, we can
easily obtain a maximum matching in $G$ in $O(|V(G)| + |E(G)| + k)$ 
time.
We will use vertex splits to ensure that the 
resulting graphs have bounded degree
and a vertex set of a certain cardinality. 
Note that if we perform a vertex split 
at $v$ in a graph of $\GG_\rho$, in general we obtain a graph of 
$\GG_{\rho+2}$ because we can represent it by making two new 
copies of the object corresponding
to $v$. Nevertheless, this increase in the density will not be
problematic in our algorithm.

\section{Maximum matching in low-density geometric intersection graphs}
\label{sec:low_density}

\subsection{Separators and separator trees}
A graph $G$ has a $(k, \alpha)$-separation if 
$V(G)$ can be partitioned into three pairwise 
disjoint sets $X,Y,Z$ such that 
$|X\cup Z|\le \alpha |V(G)|$,
$|Y\cup Z|\le \alpha |V(G)|$,
$|Z|\le k$,
and such that there is no edge with one 
endpoint in $X$ and one endpoint in $Y$.
We say that $Z$ \emph{separates} $X$ and $Y$.
At the cost of making the constant $\alpha$ larger, 
we can restrict our attention to graphs of a 
certain minimum size.

Theorem~\ref{thm:separator} gives a 
$(c\sqrt{\rho n},\alpha')$-separation for 
every graph of $\GG_\rho$, for some constant 
$\alpha' < 1$. (A separator in $G_\UU$ is a separator
in every subgraph of $G_\UU$ that is obtained
by omitting edges from $G_\UU$.)
Furthermore, such a separation can be computed 
in expected linear time,
if the objects defining the graph are available.

A recursive application of separations can be 
represented as a binary rooted tree.
We will use so-called (weak) \emph{separator trees}, 
where the separator does not go into the 
subproblems. In such a tree, we store the 
separator at the root and recurse on each side 
to obtain the subtrees. We want to have small 
separators and balanced partitions at 
each level of the recursion, and we finish 
the recursion when we get to 
problems of a certain size.
This leads to the following definition.
Let $\gamma > 0$, $0 < \beta < 1$, and 
$0 < \alpha<1$ be constants.
We say that a graph $G$ has a \DEF{$(\gamma,\beta,\alpha)$-separator tree}
if there is a rooted binary full tree $T$ with the following properties:
\begin{itemize}
	\item(i) Each node $t \in T$ is associated with some 
	                set $Z_t \subseteq V(G)$.
	\item(ii) The sets $Z_t$, $t\in T$, partition $V(G)$, i.e.,
			$\bigcup_{t\in T} Z_t =V(G)$, and 
			$Z_t\cap Z_{t'}=\emptyset$, for distinct $t,t'\in T$.
	\item(iii) For each node $t \in T$, let $V_t = \bigcup_s Z_s$, where 
			$s$ ranges over the descendants of $t$ (including $t$). 
			Note that if $t$ is an internal node with 
			children $u$ and $v$,
			then $V_t$ is the disjoint union of $Z_t$, $V_u$, 
			and $V_v$. If $t$ is a leaf, then $V_t = Z_t$.
	\item(iv) For each internal node $t \in T$ with children $u$ and $v$,
			$(V_u, V_v, Z_t)$ is a 
			$(\gamma m^\beta,\alpha)$-separation for $G[V_t]$, 
			the subgraph of $G$ induced by $V_t$,
			where $m = |V_t| = |Z_t| + |V_u| + |V_v|$.
	\item(v) For each leaf $t \in T$, 
			we have $|V_t| = \Theta(\gamma^{1/(1-\beta)})$. 
			We have chosen the size so that $V_t$ is a 
			$(\gamma |V_t|^\beta,\alpha)$-separator for 
			the whole induced subgraph $G[V_t]$.
\end{itemize}

Yuster and Zwick~\cite{YusterZ07} provide an algorithm that computes 
a separator tree of some split graph for a given graph from 
an $H$-minor-free family.
As Alon and Yuster~\cite[Lemma 2.13]{AlonY13} point out, 
this algorithm actually works for any $\delta$-sparse hereditary 
graph family, as long as $\delta$ is constant.
Thus, the result applies to $\GG_\rho$.
We revise the construction to make the dependency on $\rho$ explicit. 

\begin{lemma}
\label{le:vertexsplit}
	Given a graph $G$ of $\GG_\rho$ with $n$ vertices,
	we can compute in $O(\rho n\log n)$ expected time a vertex-split 
	graph $G'$ of $G$ and a separator tree $T'$ for $G'$ with 
	the following properties:
	\begin{itemize}
	\item(i) the graph $G'$ has $\Theta(\rho n)$ vertices and edges;
	\item(ii) the maximum degree of $G'$ is at most $4$;
	\item(iii) $T'$ is a $(\gamma=O(\rho),\beta=1/2, \alpha)$-separator 
	        tree for $G'$, where $\alpha<1$ is a constant 
		(independent of $\rho$ and $n$). 
	\end{itemize}
\end{lemma}
\begin{proof}
    We adapt the construction of Yuster and Zwick~\cite[Lemma 2.1]{YusterZ07}, 
    with three main changes: First, Yuster and Zwick assume 
    $(\cdot, \alpha = 2/3)$-separations, but this specific 
    value of $\alpha$ is not needed.  Second, we make 
    ``dummy additions'' of 
    vertices to obtain a more balanced separation, 
    which we later use as a black box. 
    (Otherwise, we would get a $(\cdot, \cdot, O(1/\rho))$-separator 
    tree, and we would have to analyze the tree more carefully to obtain 
    the same final result.)
    Third, we work out the constants in the analysis to understand 
    the dependency on the density $\rho$.
	
    We proceed recursively.
    Consider an $m$-vertex graph $G \in \GG_\rho$ that
    appears during the recursion. 
    If $m \le C\rho$, 
    where $C$ is a sufficiently large constant, we 
    make a sequence of vertex splits to reduce the maximum degree 
    to three. This may increase the number of vertices.
    To ensure that this number is uniform,
    we add a (possibly empty) pendant path of even length
    until we get the maximum possible number of 
    $\Theta(\rho^2)$ vertices.
    The resulting separator tree $T'$ consists of a single node.
	
    Now suppose that $m > C \rho$.
    Using Theorem~\ref{thm:separator}, we get a 
    $(c \sqrt{\rho m},\alpha)$-separation $(X,Y,Z)$ of $G$. 
    Thus, $|Z| \le c \sqrt{\rho m}$. Yuster and 
    Zwick~\cite[Lemma 2.1]{YusterZ07} explain
    how to make vertex splits at the vertices of $Z$ and how 
    to redefine the separation so that the vertices of the separator 
    have maximum degree three. 
    See Figure~\ref{fig:split3} for how to split a vertex $v\in Z$.
    After making the vertex splits of Figure~\ref{fig:split3}
    in all vertices of $Z$, we get a split graph $G^*$ of $G$ and 
    a separation $(X^*,Y^*,Z^*)$ with
	\begin{itemize}
		\item $|X|\le |X^*|\le |X|+ |Z| \le \alpha m$, 
		\item $|Y|\le |Y^*|\le |Y|+ |Z| \le \alpha m$, 
		\item $|Z^*|\le 4\cdot |Z|+ 
		6\cdot |E(G[Z])| \le (4 + 6\rho) |Z| = O(\rho^{3/2} m^{1/2})$ 
		(by Lemma~\ref{lem:edges}),
		\item vertices of $Z^*$ have degree at most $3$ in $G^*$,
		\item $G^*[X^*]$ and $G^*[Y^*]$ are isomorphic to subgraphs 
		of $G$, i.e., $G^*[X^*]$ is isomorphic to 
		$G[X\cup Z]$ minus the edges of $G[Z]$,
		and $G^*[Y^*]$ is isomorphic to 
		$G[Y \cup Z]$ minus the edges of $G[Z]$.
	\end{itemize}
	We recurse on $G[X^*]$ and $G[Y^*]$, each 
	of which lies in $\GG_\rho$, 
    as it is (isomorphic to) a subgraph of $G$.
	In particular, the density of the graphs encountered during the 
    recursion does not increase.
	\begin{figure}
		\centering
		\includegraphics[width=\textwidth,page=3]{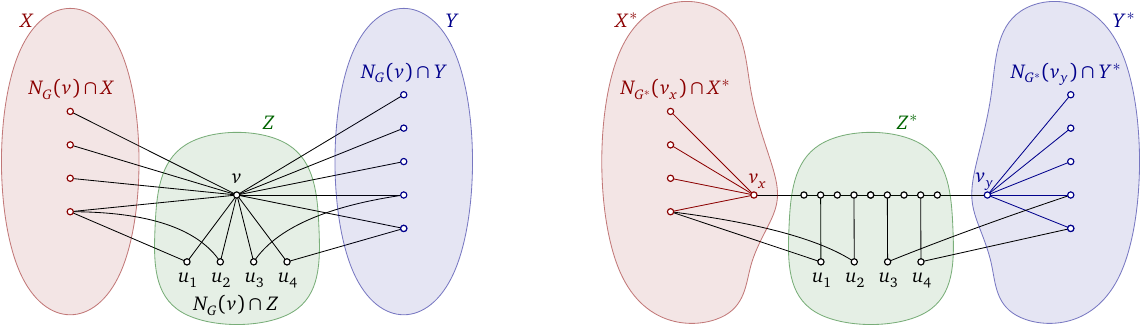}
		\caption{Splitting one single vertex of $Z$.}
		\label{fig:split3}
	\end{figure}
	The recursive call on $G^*[X^*]$ yields a graph $G'_X$ and 
    a tree $T'_X$,
	and the recursive call on $G^*[Y^*]$ yields a graph $G'_Y$ 
    and a tree $T'_Y$, both with the properties given in the theorem. 
	Let $G'$ be the graph obtained
	by putting together $G^*[Z^*]$, $G'_X$, and $G'_Y$.
	If some vertex degree gets larger than four, we can make a 
    vertex split there.
	There are at most $|Z^*|$ such vertex splits.
	A separator tree $T'$ for $G'$ is constructed by making a 
	root for $Z^*$
	and making the roots of $T'_X$ and $T'_Y$ its two children.
	Adding a pendant path of even length to $G^*[Z^*]$, if needed,
	we ensure that $|Z^*| = \Theta(\rho |Z|)$, the 
	maximum number of possible vertices after the splits 
	(we again denote the resulting vertex set by $Z^*$).
	
	For a graph $G$ with $m$ vertices considered during the recursion, 
    by Theorem~\ref{thm:separator},
	we spend $\Theta(m)$ expected time to find the separation $(X,Y,Z)$.
	Then, we construct the induced graph $G[Z]$ in 
	$O(\rho |Z|) = 
    O(\rho^{3/2} m^{1/2})$ time. 
    The transformation 
    from $G$ to $G^*$
	can be done in $O(\rho |Z|)=O(\rho^{3/2} m^{1/2})$ time. 
	Finally, when we may add vertices to $|Z^*|$, we spend 
	$\Theta(\rho m)$ time.  Standard tools to analyze 
	recursions imply that the expected running time is 
    distributed evenly over the $O(\log n)$ levels of the tree.
    Thus, the expected total running time is
	$O(\rho n \log n)$. 
	
	It is easy to see that the number of vertices of $G'$ 
	is $\Theta(\rho m)$
	because $G'_X$ has $\Theta(\rho |X^*|)$ vertices,
	$G'_Y$ has $\Theta(\rho |Y^*|)$ vertices,
	and $G[Z^*]$ has $\Theta(\rho |Z|)$ vertices.
	Since $G'$ has bounded maximum degree, it also has 
	$\Theta(\rho m)$ edges.
	Furthermore, $(V(G'_X),V(G'_Y),Z^*)$ is a separation of $G'$.
	Since $G'$ has $\Theta(\rho m)$ vertices, 
	$G'_X$ has $\Theta(\rho |X^*|)=\Theta(\rho |X|)$ vertices,
	and $(X,Y,Z)$ is a separation of $G$, 
	$|V(G'_X)| \le \alpha' m$ for some constant $\alpha'<1$. 
	The same argument applies to $|V(G'_Y)|$.
	Since 
	\[
		|Z^*|=\Theta(\rho |Z|)= O(\rho \sqrt{\rho m}) 
			 =\Theta(\rho \sqrt{|V(G')|}),
	\]
	it follows that $T'$ is a 
	$(O(\rho),\beta=1/2, \alpha')$-separator tree.
\end{proof}

Note that the split graph $G'$ 
in Lemma~\ref{le:vertexsplit}
is not necessarily in $\GG_\rho$. 
It is a subgraph of an intersection graph, but since we introduce
copies of geometric objects when we split vertices, the density increases. 
In any case, this does not matter because $G'$ will be accessed through 
the separator tree $T'$.

\subsection{Nested dissection}
We will need to compute with 
matrices.
The arithmetic operations take place 
in $\ZZ_p$, 
where $p = \Theta(n^4)$ is prime. Thus, we 
work with numbers of $O(\log n)$-bits, and assuming
a standard unit-cost model,
we simply need to bound the number of arithmetic operations.

Let $A$ be an $n\times n$ matrix.
A Gaussian elimination step on row $i$ is the following operation:
for $j=i+1,\dots,n$,
add an appropriate multiple of row $i$ to row $j$ so that the 
element at position $(j,i)$ becomes $0$. Elimination on row $i$ can 
be performed if the entry at position $(i,i)$ is nonzero.
Gaussian elimination on $A$ consists of performing Gaussian 
elimination steps on rows $i=1,\dots, n-1$. 
This is equivalent to computing an $LU$ decomposition of $A$, 
where $L$ is a lower triangular matrix with units along the diagonal, and $U$
is an upper triangular matrix. Gaussian elimination is 
performed \emph{without pivoting} if,
for all $i$, when we are about to do a Gaussian elimination step on row $i$,
the entry at position $(i,i)$ is non-zero. 
If Gaussian elimination is performed without pivoting, then the 
matrix is non-singular.
(Pivoting is permuting the rows to ensure that 
the entry at position $(i,i)$ is non-zero.)

Let $[n]=\{1,\dots, n\}$.
The \DEF{representing graph} $G(A)$ of an $n\times n$ 
matrix $A=(a_{i,j})_{i,j\in [n]}$ is
\[
	G(A)= \left([n], \left\{ ij \in \binom{[n]}{2} 
	\left\vert a_{i,j}\neq 0 \text{ or } 
	a_{j,i}\neq 0\right.\right\} \right).
\]
Let $T$ be a separator tree for $G(A)$.
The row order of $A$ is \DEF{consistent} with $T$ if,
whenever $t'$ is an ancestor of $t$,
all the rows of $Z_t$ are before any row of $Z_{t'}$.
We may assume that all the rows of $Z_t$ are consecutive.
In particular, if the rows are ordered according to a post-order
traversal of $T$, then the row order of $A$ is consistent with $T$.
A careful but simple revision of the nested dissection algorithm 
by Gilbert and Tarjan~\cite{GilbertT86} leads to the following theorem.

\begin{theorem}
\label{thm:dissection}
	Let $A$ be an $n\times n$ matrix such that 
	the representing graph $G(A)$ has bounded degree
	and assume that we are given a $(\gamma,\beta,\alpha)$-separator 
	tree $T$ for $G(A)$,
	were $\gamma>0$, $0< \alpha <1$, and $1/2 \leq \beta < 1$ are constants.
	Furthermore, assume that the row order of $A$ is consistent with $T$
	and that Gaussian elimination on $A$ is done without pivoting.
	We can perform Gaussian elimination (without pivoting) on $A$
	and find a factorization $A=LU$ of $A$ in 
	$O(\gamma^\omega n^{\beta \omega})$ time, where
	$L$ is a lower triangular matrix with units along the diagonal 
	and $U$ is an upper triangular matrix. 
\end{theorem}

For the proof of Theorem~\ref{thm:dissection}, we will need the 
following folklore lemma, whose proof we include for completeness.

\begin{lemma}
\label{le:partialelimination}
	Let $A$ be an $n \times n$ matrix, and 
	$k \le n$. Suppose that Gaussian elimination on
	the first $k$ rows of $A$ needs no pivoting. Then, 
	we can perform Gaussian elimination on the 
	first $k$ rows of $A$
	with $O(n^2 k^{\omega-2})$ arithmetic operations.
\end{lemma}
\begin{proof}
	Computing the inverse or performing Gaussian elimination 
	for a $k\times k$ matrix takes $O(k^\omega)$ time 
	(even if pivoting is needed), see, e.g.,
	Bunch and Hopcroft~\cite{BunchH74}, and Ibarra, Moran, 
	and Hui~\cite{IbarraMH82}. 
\def\unitk{1.3cm}\relax
\def\unitn{3.3cm}\relax
Assume that 
	\[
		A~=~
		\left[ 
		\begin{array}{c|c}
			\vbox to \unitk {\vfill \hbox to \unitk{\hfill $A_{1,1}$	\hfill} \vfill} & 
			\vbox to \unitk {\vfill \hbox to \unitn{\hfill $A_{1,2}$	\hfill} \vfill} 
			\\ \hline
			\vbox to \unitn {\vfill \hbox to \unitk{\hfill $A_{2,1}$	\hfill} \vfill} & 
			\vbox to \unitn {\vfill \hbox to \unitn{\hfill $A_{2,2}$	\hfill} \vfill} 
		\end{array}
		\right] ~,
	\]
	where $A_{1,1}$ is $k\times k$ and $A_{2,2}$ is $(n-k)\times (n-k)$.
	We want to perform Gaussian elimination without pivoting for 
	the first $k$ rows. 
	
	First, we perform Gaussian elimination on the $k\times k$ matrix $A_{1,1}$. 
	This takes $O(k^\omega)$ time, and we obtain two $k\times k$ matrices $L$ and $U$ 
	such that $A_{1,1}=LU$, 
	the matrix $L$ is lower triangular with units along the diagonal,
	and the matrix $U$ is upper triangular. 
	Since we use no pivoting, $L$ and $U$ are non-singular.
	We also compute in $O(k^\omega)$ time the inverses 
	$(A_{1,1})^{-1}$, $L^{-1}$, and $U^{-1}$.
	Then, we have 
	\[
		A~=~
		\left[ 
		\begin{array}{c|c}
			\vbox to \unitk {\vfill \hbox to \unitk{\hfill $L$					\hfill} \vfill} & 
			\vbox to \unitk {\vfill \hbox to \unitn{\hfill $0_{k\times (n-k)}$	\hfill} \vfill} 
			\\ \hline
			\vbox to \unitn {\vfill \hbox to \unitk{\hfill $A_{2,1}U^{-1}$		\hfill} \vfill} & 
			\vbox to \unitn {\vfill \hbox to \unitn{\hfill $I_{(n-k)\times (n-k)}$		\hfill} \vfill} 
		\end{array}
		\right] 
		\left[ 
		\begin{array}{c|c}
			\vbox to \unitk {\vfill \hbox to \unitk{\hfill $U$					\hfill} \vfill} & 
			\vbox to \unitk {\vfill \hbox to \unitn{\hfill $L^{-1}A_{1,2}$	\hfill} \vfill} 
			\\ \hline
			\vbox to \unitn {\vfill \hbox to \unitk{\hfill $0_{(n-k)\times k}$		\hfill} \vfill} & 
			\vbox to \unitn {\vfill \hbox to \unitn{\hfill $A_{2,2}-A_{2,1} (A_{1,1})^{-1} A_{1,2}$ \hfill} \vfill} 
		\end{array}
		\right] ~,
	\]
	which means that the second matrix on the right side is the 
	result of making
	Gaussian elimination for the first $k$ rows of $A$. The products 
	$A_{2,1}U^{-1}$, $L^{-1}A_{1,2}$,  and $A_{2,1} (A_{1,1})^{-1} A_{1,2}$
	can be computed in $O(n^2 k^{\omega-2})$ time by making at most 
	$O(n^2/k^2)$ products of submatrices of size $k\times k$.
\end{proof}

\begin{proof}[Proof of Theorem~\ref{thm:dissection}]
	We assume that the reader is familiar with some of the previous 
	work on nested dissection to compute elimination orders for 
	Gaussian elimination~\cite{GilbertT86,LiptonRE79}.
	Set $G = G(A)$. For each edge $ij$ of $G$, if $i \in Z_t$ 
	and $j \in Z_{t'}$, 
	then either $t = t'$ or $t$ and $t'$ have an 
	ancestor-descendant relation in $T$.
	
	For each node $t$ of $T$, we eliminate all rows in $Z_t$ together,
	using block Gaussian eliminations. 
	Since the row order is consistent with $T$, we have already 
	eliminated all the rows
	of $V_t\setminus Z_t$, and we have not yet eliminated any 
	row of $Z_{t'}$, for any ancestor $t'$ of $t$. 
	
	For each node $t$ of $T$, let $B_t$ be the set of vertices $j$ that
	belong to some $Z_{t'}$, where $t'$ is an ancestor of $t$ in $T$, 
	such  that there is an edge from $j$ to some vertex of $V_t$. 
	A vertex $j$ of $G(A)$ is affected by the elimination steps on 
	the rows of $Z_t$ only if $j$ belongs to $Z_t$ or to $B_t$.
	Thus, performing Gaussian elimination steps on the rows of 
	$Z_t$ affects at most $|Z_t|+|B_t|$ rows and columns.
	Eliminating the rows of $Z_t$ affects the rows of $B_t$. However, 
	when processing node $t$,
	we do not yet perform any elimination steps on the rows of $B_t$.
	Thus, we consider the submatrix with indices in $Z_t \cup B_t$,
	and we perform the elimination steps only on the rows of $Z_t$. 
	By Lemma~\ref{le:partialelimination},
	this takes $O((|Z_t|+|B_t|)^2 |Z_t|^{\omega-2})$ time. 
	It follows that the running time of the whole algorithm is 
	\begin{equation}
		\sum_{t\in T}~~ O\big((|Z_t|+|B_t|)^2 |Z_t|^{\omega-2}\big) 
		~~=~~ 
		\sum_{t\in T}~~ O(|Z_t|^\omega) + \sum_{t\in T}~~ 
		O(|B_t|^2 |Z_t|^{\omega-2}).\label{eq:1}
	\end{equation}
	Since $|Z_t|\le \gamma |V_t|^\beta$, the first sum is bounded
	as follows:
	\begin{equation}
		\sum_{t\in T} |Z_t|^\omega ~~\leq~~ 
		\sum_{t\in T} \gamma^\omega
                |V_t|^{\beta\omega} ~~=~~
		\gamma^\omega \cdot \sum_{t\in T} |V_t|^{\beta\omega} ~~=~~
		O(\gamma^\omega \cdot n^{\beta\omega} ),
	\label{eq:2}
	\end{equation}
	where in the last step we have used the assumption $\beta\omega>1$.

	To bound the second sum, we first
	analyze $\Sigma = \sum_{t\in T} |B_t|^2$.
	For this, we follow Gilbert and Tarjan~\cite{GilbertT86} almost 
	verbatim.  Let $L_\ell$ be the nodes of $T$ at level $\ell$ and define 
	$\Sigma_\ell = \sum_{t\in L_\ell} |B_t|^2$. (The root is at level $0$.)
	Fix a level $\ell>0$.  The sum $\Sigma_\ell$ is maximized
        if for each node $t'$ at level $\ell'<\ell$, all the edges 
        with an endpoint in $t'$ and an endpoint at level at least 
        $\ell$ are incident to the same subgraph $G[V_t]$ of $t\in L_\ell$.
        That is, to bound $\Sigma_\ell$, we can assume that all the 
	edges incident to $Z_{t'}$ contribute to the same $B_t$, $t\in L_\ell$.
	For each $t\in L_ \ell$, let $s(t)$ be the highest node of 
	$T$ with an edge going to $V_t$.  Because of the assumption 
	we made, the mapping $t \mapsto s(t)$ (from $L_\ell$ 
	to $\cup_{\ell'\le \ell} L_{\ell'}$) is injective. 
	If $s(t)=t_0,t_1,\dots, t_a= t$ is the path in $T$ from $s(t)$ to $t$,
	then we have $|V_{t_i}|\le \alpha^i |V_{s(t)}|$ for each $i=0,\dots,a$.
	Using that each vertex of each $Z_{t_i}$ has bounded degree,
	we get  
	\begin{align*}
          |B_t| ~~&\le~~ \sum_{i=0}^a O(1)\cdot |Z_{t_i}| 
                 ~~\le~~  O(1)\cdot \sum_{i=0}^a \gamma |V_{t_i}|^\beta  
	~~\le~~  O(\gamma)\cdot \sum_{i=0}^a (\alpha^i |V_{s(t)}|)^\beta \\
	~~&\le~~~  O(\gamma \cdot |V_{s(t)}|^\beta) \cdot \sum_{i=0}^a (\alpha^\beta)^i 
			~~\le~~  O\left(\gamma \cdot |V_{s(t)}|^\beta \cdot 
			  \frac{1}{1-\alpha^\beta} \right).
	\end{align*}
	Since the map $t\mapsto s(t)$ is an injection 
	(when $t\in L_\ell$), we have
	\begin{align*}
		\Sigma_\ell~~=~~
		\sum_{t \in L_\ell} |B_t|^2 ~~&\le~~ 
		\sum_{t \in L_\ell} 
		O\left( \gamma^2 \cdot |V_{s(t)}|^{2\beta} \cdot
		  \frac{1}{(1-\alpha^\beta)^2}\right)\\ 
		~~&\le~~ 
		\sum_{s\in \cup_{\ell'\le\ell}L_{\ell'}} O\left( \gamma ^2 \cdot |V_s|^{2\beta} \cdot \frac{1}{(1-\alpha^\beta)^2}\right)\\ ~~&=~~
		O\left( \gamma^2\cdot \frac{1}{(1-\alpha^\beta)^2} \cdot \ell\cdot n^{2\beta}\right),
	\end{align*}
	where in the last step we have used that the sets $V_s$, 
	$s\in L_{\ell'}$, are pairwise disjoint subsets of $[n]$ 
	for each level $\ell'$, and $2\beta\ge 1$.
	For each $\ell$ and each $t\in L_\ell$, we have $|V_t|\le \alpha^\ell n$
	and therefore $|Z_t|\le \gamma (\alpha^\ell n)^\beta$.
	This implies that
	\begin{align*}
		\sum_{t\in L_\ell} |B_t|^2 |Z_t|^{\omega-2} ~~&\le~~
		\sum_{t \in L_\ell} |B_t|^2 \Bigl( \gamma \bigr( \alpha^\ell n \bigr)^\beta \Bigr)^{\omega-2} ~~=~~  
		\gamma^{\omega-2} \cdot (\alpha^\ell n)^{\beta(\omega-2)} \cdot \sum_{t \in L_\ell} |B_t|^2 \\
		~~&\le~~ 
		\gamma^{\omega-2} \cdot (\alpha^\ell n)^{\beta(\omega-2)} \cdot O\left( \gamma^2\cdot \frac{1}{(1-\alpha^\beta)^2} \cdot \ell\cdot n^{2\beta}\right)\\ 
		~~&=~~
		O\left(\gamma^{\omega} \cdot \frac{1}{(1-\alpha^\beta)^2} \cdot n^{\beta\omega} \cdot \alpha^\ell \ell\right).
	\end{align*}
	Since $\sum_{\ell\ge 0} \alpha^\ell \ell = \alpha/(1-\alpha)^2$, for  
	$0<\alpha<1$, we get that 
	\begin{align*}
		\sum_{t\in T}~~ |B_t|^2 |Z_t|^{\omega-2} &=
		\sum_{\ell\ge 0} \sum_{t \in L_\ell} |B_t|^2 |Z_t|^{\omega-2}\\
		& \le 
		\sum_{\ell\ge 0} O\left(\gamma^{\omega} \cdot 
		\frac{1}{(1-\alpha^\beta)^2}  \cdot n^{\beta\omega} \cdot \alpha^\ell \ell\right)\\ ~~&=~~
		O\left(\frac{\alpha}{(1-\alpha^\beta)^2(1-\alpha)^2}\cdot\gamma^{\omega} \cdot n^{\beta\omega}\right).
	\end{align*}
	Combining it with \eqref{eq:2}, we get from 
	\eqref{eq:1} that the total running time
	is $O((\alpha/(1-\alpha^\beta)(1-\alpha)^2)\cdot\gamma^{\omega} \cdot n^{\beta\omega})$.
	The theorem follows (in the statement of our theorem, we hide
	$\alpha$ in the $O$-notation, to avoid clutter and since the precise
	dependency is not important in our applications).
\end{proof}

\noindent\textbf{Remark 1:}
Mucha and Sankowski~\cite{MuchaS06} noted that the result holds 
when $G(A)$ is planar or, more generally, has recursive separators,
using the approach by Lipton, Rose, and Tarjan~\cite{LiptonRE79} 
for nested dissection.  This approach is based on the \emph{strong} 
separator tree.  Alon, Yuster, and Zwick~\cite{AlonY13,YusterZ07} 
showed that a similar result holds for graphs of bounded degree with 
recursive separators if one instead uses the nested
dissection given by Gilbert and Tarjan~\cite{GilbertT86}. In this 
case, we need bounded degree, but a weak separator tree suffices.
Again, since we want to make the dependency on $\rho$ explicit 
and since the analysis in terms of matrix multiplication time does 
not seem to be written down in detail anywhere,
we revise the method carefully.

\noindent\textbf{Remark 2:}
Usually, the result is stated for symmetric positive definite matrices.
Reindexing a symmetric positive definite matrix gives another symmetric positive definite matrix,
and Gaussian elimination on such matrices can always be performed without pivoting.
Thus, for positive semidefinite matrices, we do not need to assume that 
the row order is consistent with $T$ because we can reorder the rows to make
it consistent with $T$. However, Mucha and Sankowski~\cite{MuchaS06} 
do need the general statement in their Section 4.2, 
and they mention this general case after their Theorem~13. Actually, 
they need it over $\ZZ_p$,
where the concept of positive definiteness is not even defined!

\subsection{The algorithm}
Assume we have a graph $G$ of $\GG_\rho$ with $n$ vertices
and a geometric representation, i.e., geometric objects $\UU$
of density at most $\rho$ such that $G$ is a subgraph of $G_\UU$. 
We want to compute a maximum matching for $G$.
For this, we adapt the algorithm of Mucha and Sankowski~\cite{MuchaS06}.
We provide an overview of the approach, explain the necessary modifications,
and emphasize the dependency on $\rho$ in the different
parts of the algorithm.

Using Lemma~\ref{le:vertexsplit}, we get in $O(\rho n\log n)$ 
expected time a vertex-split graph $G'$ of $G$ and a separator tree $T'$
for $G'$ such that 
\begin{itemize}
\item(i) the graph $G'$ has $\Theta(\rho n)$ vertices and edges;
\item(ii) the maximum degree of $G'$ is at most $4$;
\item(iii) $T'$ is a $(\gamma=O(\rho),\beta=1/2, \alpha)$-separator 
        tree for $G'$,
	where $\alpha<1$ is a constant (independent of $\rho$ and $n$). 
\end{itemize}
Since $G'$ is obtained from $G$ by vertex splits, it 
suffices to find a maximum matching in $G'$.
We set $m=|V(G')|=\Theta(\rho n)$, and we label
the vertices of $G'$ from $1$ to $m$.
We consider the variables $X=(x_{ij})_{ij\in E(G')}$; i.e.,
each edge $ij$ of $G$ defines a variable $x_{ij}$. 
Consider the $m\times m$ symbolic matrix $A[X]=A[X](G')$, defined as follows:
\[
	(A[X])_{i,j} = \begin{cases}
				x_{ij}, & \text{if $ij\in E(G')$ and $i<j$},\\
				-x_{ij}, & \text{if $ij\in E(G')$ and $j<i$},\\
					0 & \text{otherwise}.
   		\end{cases}
\]
The symbolic matrix $A[X]$ is usually called the \emph{Tutte matrix} of $G'$.
It is known~\cite{RabinV89} that the rank of $A[X]$ is 
twice the size of the maximum matching in $G'$. In particular, $G'$
has a perfect matching if and only if $\det(A[X])$ is not identically zero.
Take a prime $p=\Theta(n^4)$, and substitute each variable in $A[X]$
with a value from $\ZZ_p$, each chosen independently uniformly at random.
Let $A$ be the resulting matrix. 
Then, with high probability, $\rank(A)=\rank(A[X])$, where on both sides 
we consider the rank over the field $\ZZ_p$~\cite{RabinV89}.

\paragraph*{From maximum matching to perfect matching.}
Let $B=AA^T$. Then, $B$ is symmetric, and
the rank of $B$ equals the rank of $A$. 
Note that $(B)_{i,j}$ is nonzero only if $i$ and $j$ share a neighbor in $G'$.
Since $G'$ has bounded degree, 
from the separator tree $T'$ for $G'$, we can obtain a separator tree $T_B$
for the representing graph $G(B)$. Since $T'$ was a 
$(\gamma=O(\rho),\beta=1/2, \alpha)$-separator tree for $G'$, 
$T_B$ is a $(\gamma=O(\rho),\beta=1/2, \alpha)$-separator tree for $G(B)$, 
where the constant hidden in $O(\rho)$ is increased by the 
maximum degree in $G'$.
Using Theorem~\ref{thm:dissection}, we obtain that 
Gaussian elimination
can be done in $B$ in 
$O(\gamma^\omega m^{\omega/2})= O(\rho^\omega (\rho n)^{\omega/2})= 
O(\rho^{3\omega/2} n^{\omega/2})$
time, assuming that pivoting is not needed. 

Mucha and Sankowski~\cite[Section 5]{MuchaS06} show how Gaussian elimination
without pivoting can be used in $B$ to find a collection 
of indices $W\subseteq [m]$ such
that the centered matrix $(B)_{W,W}$, 
defined by rows and columns of $B$ with indices in $W$,
has the same rank as $B$. It follows that $\rank(A_{W,W})=\rank(B_{W,W})$
and therefore $G'[W]$ contains a maximum matching of $G'$ that is
a perfect matching in $G'[W]$ (with high probability). 
The key insight to find such $W$ is that, if during
Gaussian elimination in $B$ we run into a $0$ along the diagonal,
then the whole row and column are $0$, which means that they can be removed
from the matrix without affecting the rank.
We summarize. 

\begin{lemma}
	In time $O(\rho^{3\omega/2} n^{\omega/2})$ we can find a subset $W$
	of vertices of $G'$ such that, with high probability,
	$G'[W]$ has a perfect matching that is a maximum matching in $G'$.
\end{lemma}

From now on, we can assume that $G'$ has a perfect matching.
We keep denoting by $T'$ its separator tree, by $A$ the matrix after
substituting values of $\ZZ_p$ into $A[X]$, and by $B$ the matrix $AA^T$. 
(We can compute the tree $T'$ anew or we can reuse the same separator 
tree restricted to the subset of vertices.)
Let $Z_r$ denote the set stored at the root $r$ of $T'$.
Thus, $Z_r$ is the first separator on $G'$.
Let $N_r$ be the set $Z_r$ together with its neighbors in $G'$.
Because $G'$ has bounded degree, we have 
$|N_r|=O(|Z_r|)=O(\rho m^{1/2})= O(\rho^{3/2} n^{1/2})$.

Mucha and Sankowski show how to compute with $O(1)$ Gaussian
eliminations a matching $M'$ in $G'$ that covers all the vertices
of $Z_r$ and is contained in some perfect matching of $G'$. 
There are two ingredients for this. The first ingredient is to
use Gaussian elimination on the matrix $B = AA^T$ to obtain
a decomposition $AA^T=LDL^T$, and then use (partial) Gaussian elimination
on a matrix $C$ composed of $L_{[m],N_r}$ and $A_{N_r,[m]\setminus N_r}$
to compute $(A^{-1})_{N_r,N_r}$. (Note that in general $(A^{-1})_{N_r,N_r}$
is different from $(A_{N_r,N_r})^{-1}$. Computing the latter is simpler,
while computing the former is a major insight by 
Mucha and Sankowski~\cite[Section 4.2]{MuchaS06}.)
Interestingly, $T'$ is also a separator tree for 
the representing graph of this matrix $C$, and Gaussian elimination
can be performed without pivoting. Thus, we can obtain in 
$O(\rho^{\omega}m^{\omega/2})= O(\rho^{3\omega/2} n^{\omega/2})$ time
the matrix $(A^{-1})_{N_r,N_r}$.
The second ingredient is that, once we have $(A^{-1})_{N_r,N_r}$, 
we can compute for any matching $M'$ contained in 
$G'[N_r]$ a maximal (with respect to inclusion) 
submatching $M'$ that is contained 
in a perfect matching of $G'$.
This is based on an observation by Rabin and Vazirani~\cite{RabinV89} 
that shows how to find edges that belong to some perfect 
matching using the inverse
matrix, and Gaussian elimination on the matrix $(A^{-1})_{N,N}$
to identify subsets of edges that together belong to some perfect matching.
The matrix $(A^{-1})_{N_r,N_r}$ is 
not necessarily represented by a graph with nice separators,
but it is of size $|N_r|\times |N_r|$.
Thus, Gaussian elimination in $(A^{-1})_{N_r,N_r}$ takes
$O(|N_r|^\omega)=O(\rho^{3\omega/2} n^{\omega/2})$ 
time~\cite[Section 2.4]{MuchaS06}.

Since the graph $G'$ has bounded maximum degree, making $O(1)$ iterations
of finding a maximal matching $M'$ in $G'[N_r]$, followed by finding a maximal
subset $M''$ of $M'$ contained in a perfect matching of $G'$, and removing
the vertices contained in $M'$ plus the edges of $M'\setminus M''$,
gives a matching $M_*$ that covers $Z_r$ and 
is contained in a perfect matching of $G'$; see~\cite[Section 4.3]{MuchaS06}.
The vertices of $M_*$ can be removed, and we recurse on both
sides of $G' - V(M_*)\subset G' - Z_r$ using the corresponding subtrees of $T'$.
The running time is $T(n)= O(\rho^{3\omega/2} n^{\omega/2}) + T(n_1)+T(n_2)$,
where $n_1,n_2\le \alpha n$. This solves 
to $T(n) = O(\rho^{3\omega/2} n^{\omega/2})$
because $\omega/2>1$. 
We summarize in the following result.
If only the family $\UU$ is given,
first we use Lemma~\ref{lem:edges} to construct $G_\UU$.

\begin{theorem}
\label{thm:main1}
	Given a graph $G$ of $\GG_\rho$ with $n$ vertices 
	together with a family $\UU$ of geometric objects 
	with density $\rho$ such that $G$ is a subgraph of $G_\UU$,
	we can find in $O(\rho^{3\omega/2} n^{\omega/2})$ time 
	a matching in $G$ that, with high probability, is maximum.
       In particular, for a family $\UU$ of $n$ geometric 
       objects with density $\rho$, a maximum matching in $G_\UU$ 
       can be found in $O(\rho^{3\omega/2} n^{\omega/2})$ time.  
       The same holds for the bipartite or $k$-partite version of $G_\UU$.
\end{theorem}


\section{Sparsification}
\label{sec:sparsification}

Let $\UU$ be a family of convex geometric objects in the plane
such that each object of $\UU$ contains a square 
of side length $1$ and is contained in a square 
of side length $\Psi\ge 1$.
Through the discussion we will treat $\Psi$ as a parameter.
Our objective is to reduce the problem of computing
a maximum matching in the intersection graph $G_\UU$
to the problem of computing a maximum matching in
$G_\WW$ for some $\WW\subseteq \UU$ of small depth.

Let $P = \ZZ^2$ be the points in the plane with integer coordinates.
Each square of unit side length contains at least one 
point of $P$ and each square of side length $\Psi$ contains
at most $(1+\Psi)^2=O(\Psi^2)$ points of $P$.
In particular, each object $U\in \UU$ contains at least
one and at most $O(\Psi^2)$ points from $P$.

First we provide an overview of the idea.
The objects intersected by a point $p \in P$
define a clique, and thus any even number of them defines 
a perfect matching. We show that, for each $p \in P$, 
it suffices to keep a few objects pierced by $p$, and we show
how to obtain such a suitable subfamily. The actual number
of objects to keep depends on $\Psi$, and whether the actual
computation can be done efficiently 
depends on the geometric shape of the objects.

For each object $U \in \UU$, we find the lexicographically
smallest point in $P\cap U$. We assume that we have a primitive
operation to compute $P \cap U$ for each object $U\in \UU$
in $O(1+ |P\cap U|)= O(\Psi^2)$ time.
A simple manipulation of these incidences
allows us to obtain the \DEF{clusters}
\[
	\UU_p ~=~ \{ U\in \UU \mid \text{$p$ lexicographically minimum
								in $P\cap U$}\},
	~~~\text{ for all $p\in P$}.
\]
Note that the clusters $\UU_p$, for $p\in P$,
form a partition of $\UU$. This will be useful later.
Clearly, the subgraph of $G_\UU$ induced
by $\UU_p$ is a clique, for each $p\in P$.

We will use the usual notation
\[
	E(\UU_p,\UU_q) ~=~ \{ UV\mid U\in \UU_p, \, V\in \UU_q, 
								U\cap V\neq \emptyset \}
					~\subseteq~ E(G_\UU).
\]
The \DEF{pattern graph} $H=H(P,\Psi)$ 
has vertex set $P$ and set of edges 
\[
	E(H)~=~\{ pq \mid 
			\|p-q\|_\infty \le 2\Psi\}
	~\subseteq~\binom{P}{2}.
\]
The use of the pattern graph is encoded in the following
property: if $U\in \UU_p$, $V\in \UU_q$ 
and $U\cap V\neq \emptyset$, then $pq\in E(H)$.
Indeed, if $U$ and $V$ intersect, then the union $U\cup V$
is contained in a square of side length $2\Psi$,
and thus the $L_\infty$-distance between
each $p\in P\cap U$ and $q\in P\cap V$ is at most $2\Psi$.

The definition of $H(P, \Psi)$ implies that the edge set of $G_\UU$ is 
the disjoint union of $E(\UU_p,\UU_q)$, over all $pq\in E(H)$,
and the edge sets of the cliques $G_{\UU_p}$, over all $p\in P$.
Thus, whenever $pq\notin E(H)$, there are no edges in $E(\UU_p,\UU_q)$.

Let $\lambda$ be the maximum degree of $H$. 
Note that $\lambda=O(\Psi^2)$.
The value of $\lambda$
is an upper bound on how many clusters $\UU_q$ may interact with
a single cluster $\UU_p$. We will use $\lambda$
as a parameter to decide how many objects from each $\UU_p$
are kept. 
We start with a simple observation.

\begin{lemma}
\label{le:one_edge1}
	There exists a maximum matching in $G_\UU$ that,
	for all $pq\in E(H)$, 
	contains at most one edge of $E(\UU_p,\UU_q)$.
\end{lemma}
\begin{proof}
	Let $M$ be a maximum matching in $G_\UU$ such 
	that $\sum_{pq \in E(H)} |M \cap E(\UU_p,\UU_q)|$ is
	minimum.
        Suppose there is an edge $p_0q_0\in E(H)$ with  
	$|M \cap E(\UU_{p_0},\UU_{q_0})|\ge 2$.
	Then we have two edges $UV$ and $U'V'$ in $M\cap E(\UU_{p_0},\UU_{q_0})$,
	where $U,U'\in \UU_{p_0}$ and $V,V'\in \UU_{q_0}$.
	Since $UU'$ and $VV'$ are also edges in $G_\UU$,
	we see that $M'=\bigl( M\setminus \{ UV, U'V'\} \bigr) \cup \{ UU', VV'\}$
	is a maximum matching in $G_\UU$.
	We then have 	
	\[
		|M'\cap E(\UU_{p_0},\UU_{q_0})| ~=~ |M\cap E(\UU_{p_0},\UU_{q_0})|-2,
	\]
	and 
	\[
		|M'\cap E(\UU_p,\UU_q)| ~=~ |M\cap E(\UU_p,\UU_q)|,
		~~~\text{for all } pq \in E(H), pq \neq p_0q_0.
	\]
	In this last statement, it is important that $\UU_p$, $p\in P$,
	is a partition of $\UU$, as otherwise $UU'$ could belong
	to some $E(\UU_p,\UU_q)$ or even $E(\UU_{p_0},\UU_{q_0})$.
	Hence, $\sum_{pq \in E(H)} |M' \cap E(\UU_p,\UU_q)|$
	is strictly smaller than 
        $\sum_{pq \in E(H)} |M \cap E(\UU_p,\UU_q)|$,
	a contradiction to our choice of $M$. The result follows.
\end{proof}

Of course we do not know which object from the cluster $\UU_p$ will
interact with another cluster $\UU_q$. 
We will explain how to get a large enough subset of cluster $\UU_p$.

For each $pq\in E(H)$, we construct a set $\WW(p,q)\subseteq \UU_p\cup\UU_q$ 
as follows.
First, we construct a matching $M=M(p,q)$ in $E(\UU_p,\UU_q)$ such that
$M$ has $2\lambda +1$ edges or $M$ has fewer than $2\lambda +1$ edges
and is maximal in $E(\UU_p,\UU_q)$. For example, such a matching 
can be constructed incrementally.
If $M$ has $2\lambda +1$ edges, we take $\WW(p,q)$ to be
the endpoints of $M$.
Otherwise, for each endpoint $U\in \UU_p$ (resp.~$V\in \UU_q$) of $M$, 
we place $U$ (resp.~$V$) and $\lambda$ of its neighbors from $\UU_q$ 
(resp.~$\UU_p$) into $\WW(p,q)$. When $U$ (resp. $V$)
has fewer than $\lambda$ neighbors, we place all its neighbors
in $\WW(p,q)$. This finishes the description of $\WW(p,q)$;
refer to 
Algorithm $\textit{Sparsify-one-edge}$ in 
Figure~\ref{fig:sparsify_one_edge} for pseudo-code.

\begin{figure}
        \centering
        \ovalbox{\begin{minipage}{.95\hsize}
        \centering
\normalsize
        \begin{minipage}{.92\hsize}
		\mbox{}\smallskip
        \begin{algorithm}{Sparsify-one-edge}{
                \label{algo:sparsify_one_edge}
                \qinput{$p$, $q$, $\UU_p$ and $\UU_q$}
                \qoutput{$\WW(p,q)$}}
				$\AA_p \qlet \UU_p$\\
				$\AA_q \qlet \UU_q$\\
                \qcom{compute matching $M$}\\
                $M\qlet \emptyset$\\
				\qwhile $|M|< 2\lambda+1$ \qand $\AA_p\neq \emptyset$ \qdo\\
                        $U\qlet$ an arbitrary object of $\AA_p$ \\
                        \qif $U$ intersects some $V\in \AA_q$\qthen \\
							$M \qlet M\cup \{ UV\}$\\
							$\AA_q \qlet \AA_q\setminus \{ V\}$
						\qfi\\
						$\AA_p \qlet \AA_p\setminus \{ U\}$
                \qend\\ 
                \qcom{end of computation of $M$}\\
				$\WW \qlet \cup_{UV\in M} \{ U,V \}$
					~~ \qcom{endpoints of $M$}\\
				\qif $|M|=2\lambda+1$ \qthen ~~\qcom{$M$ is a large enough matching}\\
					\qreturn $\WW$\\
				\qelse ~~\qcom{$M$ maximal but small; 
						add neighbors of $\WW$ to the output}\\
					$\mathcal{W'}\qlet \WW$\\
					\qfor $W\in \WW$ \qdo\\
						\qif $W\in \UU_p$ \qthen\\
							add to $\mathcal{W'}$ 
							$\min\{\lambda, |E(\{W\},\UU_q)|\}$
							elements of $\UU_q$ intersecting $W$\\
						\qelse ~~\qcom{ $W\in \UU_p$}\\
							add to $\mathcal{W'}$ 
							$\min\{\lambda, |E(\UU_p,\{W\})|\}$
							elements of $\UU_p$ intersecting $W$
						\qfi
					\qrof\\
					\qreturn $\mathcal{W'}$
        \end{algorithm}
        \mbox{}\smallskip
        \end{minipage}
        \end{minipage}
        }
\caption{Algorithm \emph{Sparsify-one-edge}}
\label{fig:sparsify_one_edge}
\end{figure}    

\begin{lemma}
\label{le:sparsification2}
	A maximum matching in 
	\[
		\tilde G ~=~ \left(\bigcup_{pq\in E(H)}G_{\WW(p,q)}\right)\cup 
		\left(\bigcup_{p\in P}G_{\UU_p}\right).
	\]
	is a maximum matching in $G_\UU$.
\end{lemma}
\begin{proof}
	By Lemma~\ref{le:one_edge1}, there is
	a maximum matching $M$ in $G_\UU$ 
	such that for each $pq\in E(H)$, the matching $M$
	contains at most one edge from $E(\UU_p,\UU_q)$.
	Among all such maximum matchings, we choose one matching $M$ that
	minimizes the number of edges $pq \in E(H)$ for which
	$M$ contains an edge in $E(\UU_{p},\UU_q)$ that does
	not have both vertices in $\WW(p,q)$. If there is no
	such edge $p_0q_0 \in E(H)$, then the lemma holds because 
	such $M$ is contained in $\tilde G$.
	We show that this is the only possible case.
	
	Suppose, for the sake of reaching a contradiction,
	that there exists $p_0q_0 \in E(H)$ 
	such that $M$ contains an edge $UV$ 
	with $U\in \UU_{p_0}$, $V\in \UU_{q_0}$, 
	and $\{U, V\} \not\subset \WW(p_0,q_0)$.
	Let $M'$ be the set of edges from $M$ connecting different clusters,
	i.e., $M' = M\cap \bigl( \cup_{pq\in E(H)} E(\UU_p,\UU_q)\bigr)$.
	Let $M(p_0,q_0)$ be the matching in $E(\UU_{p_0},\UU_{q_0})$
	used during the construction of $\WW(p_0,q_0)$.
	We distinguish two cases:
	\begin{itemize}
		\item $|M(p_0,q_0)| = 2\lambda + 1$. 
			Let $N_H(p)$ be the neighbors of $p$ in $H$.
			Since $M$ has at most one edge from  
			$E(\UU_{p},\UU_q)$, for each $pq\in E(H)$,
			we obtain that $M'$ has at most $\lambda$ 
			edges with an endpoint in $\UU_{p_0}$
			and at most $\lambda$ edges with
			an endpoint in $\UU_{q_0}$, as $\lambda$ 
			is the maximum degree of $H$.
			Thus, $M(p_0,q_0)$ contains at least one edge
			$U'V'$ whose endpoints are not touched by $M'$.
			We remove from $M$ the edge $UV$ and add the 
			edge $U'V'$.
			If there was some edge 
			$U'U''\in M\cap E(\UU_{p_0},\UU_{p_0})$, 
			we also replace $U'U''$ by $UU''$ in $M$. 
			If there was some edge 
			$V'V''\in M\cap E(\UU_{q_0},\UU_{q_0})$,
			we also replace $V'V''$ by $VV''$ in $M$.
		\item $|M(p_0,q_0)| \leq 2\lambda$.
			In this case, $M(p_0,q_0)$ is a maximal matching in 
			$E(\UU_{p_0},\UU_{q_0})$. In particular,
			one of $U$ or $V$ is covered by 
			$M(p_0,q_0)$,
			as otherwise we could have added $UV$ to $M(p_0,q_0)$.
			We consider the case when $U \in \UU_{p_0}$ is 
			covered by $M(p_0,q_0)$; the other case is symmetric.
			Then $U \in \WW(p_0,q_0)$, 
			$V \not\in \WW(p_0,q_0)$, and it follows 
			that $U$ has more than $\lambda$ neighbors
			in $\UU_{q_0}$. Among the at least $\lambda$ 
			neighbors of $U$
			in $\WW(p_0,q_0)$, at most
			$\lambda - 1$ are covered by edges in $M'$. (Note
			that $V$ is covered by $M'$, but $V$ is not in 
			$\WW(p_0,q_0)$.)
			This means that there is some 
			$V'\in \UU_{q_0}\cap \WW(p_0,q_0)$
			such that $V'$ is not covered by $M'$ and 
			$UV'\in E(\UU_{p_0},\UU_{q_0})$.
			We replace in $M$ the edge $UV$ by $UV'$.
			Moreover, if $V'V''$ is an edge of $M$, 
			where necessarily 
			$V''\in \UU_{q_0}$, we replace in $M$
			the edge $V'V''$ by $VV''$.
	\end{itemize}
	In both cases, we can transform the maximum matching $M$ 
	into another maximum matching
	that contains one edge with both endpoints in $\WW(p_0,q_0)$, 
	no other
	edges of $E(\UU_{p_0},\UU_{q_0})\setminus E(G_{\WW(p_0,q_0)})$,
	and the intersection of $M$ with $E(\UU_p,\UU_q)$ has
	not changed, for all $pq\in E(H)\setminus \{ p_0 q_0\}$.
	This contradicts the choice of $M$, and the lemma follows.
\end{proof}

\begin{lemma}
\label{le:depth}
	The family of objects $\WW=\cup_{pq\in E(H)}\WW(p,q)$
	has depth $O(\Psi^8)$.
\end{lemma}
\begin{proof}
	Each $\WW(p,q)$ has $O(\lambda^2)$ 
	elements, as the matching $M(p,q)$ 
	used for the construction of $\WW(p, q)$ has $O(\lambda)$
	edges, and each such edge may add $O(\lambda)$ more
	vertices to $\WW(p,q)$.
	It follows that, for each $p\in P$, the family $\WW$ contains
	at most 
	\[
		\sum_{q\in N_H(p)}|\WW(p,q)| ~\le~ 
			\lambda \cdot O(\lambda^2) ~\le~
			O(\Psi^6)
	\]
	objects from $\UU_p$. In short, $|\WW\cap\UU_p|= O(\Psi^6)$,
	for each $p\in P$.
	
	Fix a point $z \in \RR^2$.
	Let $s$ be a unit square that contains $z$ and
	whose corners lie in $P$. 
	For every object $U \in \WW$ with $z \in U$, 
	there is a square of side length $\Psi$ that contains
	$U$ and at least one corner of $s$.
	Thus, each object $U$ of $\WW$ with $z \in U$
	belongs to $\UU_p$, for some $p\in P$
	at $L_\infty$-distance at most $1+\Psi$ from $z$. 
	It follows that $z$ can only be contained in objects of $\UU_p$ 
	for $O(\Psi^2)$ points $p\in P$,
	so the depth of $z$ in $\WW$
	is at most 
	\[
		\sum_{p\in P,\, \|z -p\|_\infty\le 1 + \Psi} |\WW\cap \UU_p|
		~\le~ O(\Psi^2)\cdot O(\Psi^6) ~=~ O(\Psi^8).
	\]
	Since $z$ was arbitrary, the lemma follows.
\end{proof}

\begin{theorem}
\label{thm:sparsification}
	Let $\UU$ be a family of $n$ geometric objects in the plane
	such that each object of $\UU$ contains a square 
	of side length $1$ and is contained in a square 
	of side length $\Psi$.
	Suppose that, for any $m \in \NN$ and for any 
	$p,q \in \ZZ^2$ with
	$|\UU_p|+|\UU_q|\le m$, we can compute the
	sparsification $\WW(p, q)$ as described above in
	time 
	$T_{\rm spars}(m)$, where $T_{\rm spars}(m)=\Omega(m)$ is convex.
	In $O(\Psi^2 \cdot T_{\rm spars}(n))$ time we can reduce 
	the problem of finding
	a maximum matching in $G_\UU$ to the problem of finding
	a maximum matching in $G_\WW$ for some $\WW\subseteq \UU$
	with maximum depth $O(\Psi^8)$.
\end{theorem}
\begin{proof}
	For each $pq\in E(H)$, we find the sparsification $\WW(p,q)$. 
	Note that $\sum_{pq\in E(H)} (|\UU_p|+|\UU_q|)\le \lambda n$,
	as each $p$ contributes $\lambda$ summands.
	Hence, the computation of $\WW(p,q)$, for all 
	$pq\in E(H)$, takes time
	\[
		\sum_{pq\in E(H)}O(T_{\rm spars}(|\UU_p|+|\UU_q|)) ~=~
		O(\lambda T_{\rm spars}(n)) ~=~ O(\Psi^2 \cdot T_{\rm spars}(n)).
	\]
	Consider the family $\WW=\cup_{pq\in E(H)}\WW(p,q)$.
	By Lemma~\ref{le:depth}, the family $\WW$ has
	depth $O(\Psi^8)$.
	By Lemma~\ref{le:sparsification2}, it suffices
	to find a maximum matching in 
	\[ 
		\tilde G ~=~ \left(\bigcup_{pq\in E(H)}G_{\WW(p,q)}\right)\cup 
			\left(\bigcup_{p\in P}G_{\UU_p}\right),
	\]	
	which is a subgraph of
	\begin{equation}\label{eq:GWW}
		G_{\WW}\cup 
		\left(\bigcup_{p\in P}G_{\UU_p}\right).
	\end{equation}
	Since each $G_{\UU_p}$ is a clique and the 
	vertices of $\UU_p\setminus \WW$
	are not adjacent to any vertex outside $\UU_p$ (in the graph 
	(\ref{eq:GWW})), 
	we can just take maximum matchings within each 
	$\UU_p' = \UU_p\setminus \WW$.
	Here, we have to take care of the parity, 
	as one vertex of $\UU_p'$ may be left unmatched
	in $G_{\UU_p'}$, but may be matched to 
	some vertex of $\UU_p\cap \WW$. To handle this, 
	for each $p \in P$ such that $|\UU'_p|$ is odd, 
	we move one
	element of $\UU'_p$ to $\WW$. Thus, we can assume
	that $|\UU'_p|$ is even, for all $p \in P$. 
	The additional elements in $\WW$
	may increase the depth of $\WW$ by $O(\Psi^2)$, which is negligible.
	Now, a maximum matching (\ref{eq:GWW})
	is obtained by joining a maximum matching in $G_{\WW}$ with
	maximum matchings in $G_{\UU'_p}$, 
	$p\in P$. The maximum matchings in $G_{\UU'_p}$, $p \in P$,
	are trivial, because it is a clique on an even number of vertices.
	The result follows.
\end{proof}

Our use of properties in the plane is very mild,
and similar results hold in any space with constant
dimension.

\begin{theorem}
\label{thm:sparsification_d}
	Let $d\ge 3$ be a constant.
	Let $\UU$ be a family of $n$ geometric objects in $\RR^d$
	such that each object of $\UU$ contains a cube
	of side length $1$ and is contained in a cube 
	of side length $\Psi$.
	Suppose that, for any $m\in \NN$ and 
	for any $p,q \in \ZZ^d$ 
	with $|\UU_p|+|\UU_q|\le m$,
	we can compute the
	sparsification $\WW(p, q)$ as described above in time 
	$T_{\rm spars}(m)$, where $T_{\rm spars}(m)=\Omega(m)$ is convex.
	In $O(\Psi^d \cdot T_{\rm spars}(n))$ time we can reduce 
	the problem of finding
	a maximum matching in $G_\UU$ to the problem of finding
	a maximum matching in $G_\WW$ for some $\WW\subseteq \UU$
	with maximum depth $(1+\Psi)^{O(d)}$.
\end{theorem}
\begin{proof}
	The pattern graph $H$ can be defined for $\ZZ^d$ also using
	the $L_\infty$-metric. Such pattern graph has maximum degree 
	$O((1+\Psi)^d)=O(\Psi^d)$. Lemmas~\ref{le:one_edge1} 
	and~\ref{le:sparsification2} hold equally in this setting.
	Lemma~\ref{le:depth} holds with an upper bound of $(1+\Psi)^{O(d)}$.
	The proof of Theorem~\ref{thm:sparsification} then applies.
\end{proof}

As we mentioned in the introduction, for fat objects,
bounded depth implies bounded density; 
see Har-Peled and Quanrud~\cite[Lemma 2.7]{Har-PeledQ17}.
If a convex object contains a cube of unit side length
and is contained in a cube of side length $\Psi$, then
it is $O(1/\Psi)$-fat; see
van der Stappen et al.~\cite{StappenHO93}, where the parameter
$1/\Psi$ goes under the name of thickness.
Combining both results, one obtains that the relation
between depth and density differs by a factor of $\Psi$.
For fixed shapes, the depth and density differ by a constant factor.

\section{Efficient sparsification}
\label{sec:sparsification_eff}
Now, we implement
Algorithm $\textit{Sparsify-one-edge}$ 
(Figure~\ref{fig:sparsify_one_edge}) efficiently. 
In particular, we must perform the test in line 7 and find
the neighbors in line 19 (and the symmetric case in line 21).
The shape of the geometric objects becomes relevant for this.
First, we note that it suffices to obtain an efficient 
semi-dynamic data structure for intersection queries.

\begin{lemma}
\label{le:data_structure}
        Suppose there is a data structure 
	with the following properties:
	for any $m \in \NN$ and for any
	$p, q \in \ZZ^2 $ with
	$|\UU_p|+|\UU_q|\le m$, we 
	can maintain a set $\AA_q\subseteq \UU_q$ under 
	deletions so that, for any query $U\in \UU_p$,
	we either find some $V\in \AA_q$ with $U \cap V \neq \emptyset$
	or correctly report that no such $V$ exists.
	Let $T_{\rm con}(m)$ be the
	time to construct the data structure, 
	$T_{\rm que}(m)$ an upper bound
	on the amortized query time,
	and $T_{\rm del}(m)$ be an upper bound on the amortized deletion time.
	Then, the running time of Algorithm $\textit{Sparsify-one-edge}$
	(Figure~\ref{fig:sparsify_one_edge}) for the 
	input $(p,q,\UU_p,\UU_q)$ is
	$T_{\rm sparse}(m)=O(T_{\rm con}(m) + 
	m T_{\rm que}(m)+ \lambda^2 T_{\rm del}(m))$.	
\end{lemma}

\begin{proof}
	First, we discuss the operations in
	lines 5--10 in Algorithm $\textit{Sparsify-one-edge}$
	(Figure~\ref{fig:sparsify_one_edge}).
	We maintain $\AA_p$ as a linked list
	and $\AA_q$ in the data structure from 
	the lemma. This takes $O(T_\text{con}(m))$ time.
	Initially, $\AA_p=\UU_p$ and $\AA_q=\UU_q$.
	In each iteration of the while-loop, 
	we query with $U$ to either obtain
	some $V \in \AA_q$ intersected by $U$,
	or correctly report that no object of $\AA_q$ intersects $U$.
	If we get some $V$ intersected by $U$, we remove $V$ from $\AA_q$
	in $O(T_\text{del}(m))$ time. (Note that we have removed at most
	$2\lambda$ elements of $\UU_q$ to obtain the current $\AA_q$.)
	In either case, we remove $U$ from $\AA_p$, in $O(1)$ time.
	The running time for this part is 
	$O(mT_\text{que}(m)+ \lambda T_\text{del}(m))$.

	Next, we discuss how to do line 19
	in Algorithm $\textit{Sparsify-one-edge}$
	(Figure~\ref{fig:sparsify_one_edge}).
	We store $\AA_q=\UU_q$ in the data structure from the 
	lemma.
	For each $W\in \WW\cap \UU_p$, we repeatedly query the data structure
	to find some $V\in \AA_q$ that intersects $W$,
	and we remove this $V$ from $\AA_q$.
	We repeat this query-delete pattern in $\AA_q$ with $W$,
	until we collect $\lambda$ neighbors of $W$ or until we 
	run out of neighbors.
	Thus, the query-deletion
	pattern happens at most $\lambda$ times, for each $W$.
        Having collected the data for $W$,
	we reverse all the deletions in $\AA_q$,
	to obtain the original data structure for $\AA_q=\UU_q$,
	and we proceed to the next object of $\WW\cap \UU_p$.
	(We do not need insertions, as it suffices to
	undo the modifications that were made in the data structure.)
	In total, we repeat $O(|\WW \cap \UU_p|)=O(\lambda)$ 
	times a pattern of $O(\lambda)$	queries and deletions followed
	by a reversal of all the operations.
	Thus, the running time is
	$T_\text{con}(m)$ to construct the data and 
	$O(\lambda^2 T_\text{que}(m)+ \lambda^2 T_\text{del}(m))$ 
	to handle the operations on the data structure.
	We can assume that $\lambda^2\le m$, as otherwise
	we do not need to run the sparsification and can take directly
	the whole set of objects.
	
	Line 21	in Algorithm $\textit{Sparsify-one-edge}$
	(Figure~\ref{fig:sparsify_one_edge})
	can be done in a similar way. The rest of the algorithms 
	are elementary steps and bookkeeping.
\end{proof}

\subsection{Disks in the plane}
When  $\UU$ consists of disks in the plane,
we can use the data structure of 
Kaplan et al.~\cite{KaplanMRSS17},
with a recent improvement by Liu~\cite{Liu22},
to sparsify an edge of the pattern graph. 
This leads to the following.

\begin{proposition}
\label{prop:sparsification_disks}
	Consider a family $\UU$ of $n$ disks in the plane
	with radii in $[1,\Psi]$.
	In $O(\Psi^6 n \log^{4} n)$ expected time,
	we can reduce the problem of finding
	a maximum matching in $G_\UU$ to the problem of finding
	a maximum matching in $G_\WW$ for some subfamily
	$\WW\subseteq \UU$ of disks	with maximum depth $O(\Psi^8)$.
\end{proposition}
\begin{proof}
	Kaplan et al.~\cite{KaplanMRSS17} describe a data structure
	for additively weighted nearest-neighbor queries: 
	maintain points $A = \{ a_1, \dots, a_n\} \subseteq \RR^2$
	in the plane, where each point $a_i$ 
	has a weight $\omega_i\in \RR$ associated to it.
	The data structure can handle insertions, deletions, and
	closest point queries (for a given $x\in \RR^2$, return
	a point in $\arg\min_{a_i\in A} \omega_i+|x-a_i|$)
	in $O(\log^{4} n)$ amortized expected 
	time.\footnote{The running time 
	in \cite{KaplanMRSS17} has more logarithmic factors, but recently 
	Liu~\cite{Liu22} presented an improved construction of
	shallow cuttings that leads to the claimed result.}
	
	This data structure can be used to dynamically
	maintain a set $\AA=\{D_1,\dots, D_n\}$ of disks
	so that, for a query disk $D$, we can either report
	one disk of $\AA$ intersected by $D$ or correctly report
	that no disk of $\AA$ intersects $D$.
	Indeed, we store $\AA$ as a set $A$ of weighted points. 
	Each disk $D_i$ is represented by
	its center $a_i$ with weight equal to  its \emph{negated} radius.
	If $x$ is a point in the plane that lies outside the union of $\AA$, 
	the closest weighted point
	of $A$ gives the first disk boundary that is touched by a growing disk
	centered at $x$. If $x$ lies inside the union of $\AA$,
	the closest weighted point gives the last boundary of a disk 
	in $\AA$ that contains $x$ and that is touched by growing a disk
	around $x$. Thus, to answer a query for a disk $D$,
	we query for the weighted point $a_i\in A$ closest to 
	the center of $D$,
	and then check whether $D$ intersects $D_i$.
	Updates and queries take $O(\log^{4} n)$ amortized expected time.

	Using Lemma~\ref{le:data_structure}, we conclude that
	$T_\text{sparse}(m)= O((m+\lambda^2)\log^{4}{m})$ expected time.
	Recall that $\lambda^2=O(\Psi^4)$.
	Because of Theorem~\ref{thm:sparsification},
	we conclude that the reduction takes
	time $O(\Psi^2(n+\Psi^4)\log^{4}n)= 
	O(\Psi^6\cdot n\log^{4} n)$ expected time.	
\end{proof}

Possibly, the method can be extended to homothets of
a single object. For this one should consider
the surfaces defined by weighted distances 
in the approach of Kaplan et al.~\cite{KaplanMRSS17}.

Since the depth and the density of a family of disks are 
linearly related,
Proposition~\ref{prop:sparsification_disks} 
and Theorem~\ref{thm:main1} with $\rho=O(\Psi^8)$ imply
the following.

\begin{theorem}
\label{thm:disks}
	Consider a family $\UU$ of $n$ disks in the plane
	with radii in the interval $[1,\Psi]$.
	In $O(\Psi^6 n \log^{4} n + \Psi^{12\omega}n^{\omega/2})$ 
	expected time, we can compute 
	a matching in $G_\UU$ that, with high probability,
	is maximum.
\end{theorem}

\subsection{Translates of a fixed convex shape in the plane}
Now, suppose $\UU$ consists of translates of 
a single convex object with non-empty interior
in the plane.
With an affine transformation, we ensure that
the object is \emph{fat}: the radii of the minimum enclosing disk 
and of the maximum
enclosed disk are within a constant factor. 
Such a transformation is standard;
e.g.,~\cite[Lemma 3.2]{ahv-aemp-04}.
Thus, we may assume that $\Psi=O(1)$.
We start with a standard lemma.

\begin{lemma}
\label{le:union_pierced}
	Let $\UU$ be a family of $n$ translates of a convex object
	in the plane that are pierced by a given point $q$.
	The union of $\UU$ can be computed in $O(n\log n)$
	time.
\end{lemma}
\begin{proof}
	The boundary of two translates of the same convex object
	intersect at most twice. This means that $\UU$ is a 
	\emph{family of pseudodisks}. 
	Let $q$ be the given point that pierces all $U\in\UU$.
	We assume that $q$ belongs to the interior; otherwise it
	is necessary to make groups of objects and use $O(1)$ 
	points that intersect all the $U\in \UU$.

	Each $U \in \UU$ defines a function 
	$\delta_U\colon[0,2\pi]\rightarrow \RR$, where $\delta_U(\theta)$
	is the length of the longest segment inside $U$ with origin $q$ 
	and angle $\theta$ with some fixed axis.
	Since $q$ is in the interior of $U$, the function $\delta_U(\cdot)$
	is continuous. We can extend each function $\delta_U$ 
	to the whole $\RR$ by taking $\delta_U(\theta)=\delta(0)$,
	for $\theta\notin [0,2\pi]$. The family
	$\{ \delta_U\mid U\in \UU\}$ of totally defined functions
	is a family of \emph{pseudoparabolas}: the graphs of any two of them
	intersect at most twice.
	
	The upper envelope of a family of $n$ pseudoparabolas can
	be computed in $O(n \log n)$ time with a divide-and-conquer
	approach. First, we note that the upper envelope of $n$
	totally defined pseudoparabolas has at most $2n-1$ pieces.
	This is a standard property from the study of 
	Davenport-Schinzel sequences.
	For the algorithm, we split the family $\UU$ into two subfamilies 
	$\UU_1$ and $\UU_2$ of
	roughly the same size, recursively compute the upper envelopes $g_1$
	of $\UU_1$ and $g_2$ of $\UU_2$, 
	and then compute the upper envelope of $g_1$ and $g_2$.
	If the upper envelopes are given as $x$-monotone curves,
	then the upper envelope of $g_1$ and $g_2$, which is
	the upper envelope of $\UU$, is obtained in additional linear time.
	Since the merging step takes linear time, the whole algorithm
	takes $O(n \log n)$ time.
	
	The maps $\delta_U$ do not need to be computed explicitly and
	the whole algorithm can actually be carried out with a rotational
	sweep around $q$. The transformation to consider the functions
	$\delta_U$ helps to bring it to familiar ground in computational
	geometry.
\end{proof}

We will use the following lemma to ``simulate'' deletions.
For this, we will keep a half-infinite interval of indices that contains the
elements that are ``deleted''.

\begin{lemma}
\label{le:data_structure_translates_1}
	Let $\UU=\{ U_1,\dots U_n\}$ be a family of $n$ translates 
	of a convex object in the plane 
	that are pierced by a given point $q$.
	In $O(n\log^2 n)$ time, we can construct a data structure
	for the following queries:
	given $x \in \RR^2$ and a value $a \in \{1, \dots, n\}$, 
	find the smallest $i\ge a$ such
	that $U_i$ contains $x$, or correctly report
	that $x$ does not belong to $U_a\cup\dots\cup U_n$.
	The query time is $O(\log^2 n)$. 
\end{lemma}
\begin{proof}
    We follow the standard approach for adding
	range capabilities to data structures~\cite{WillardLu85}:
	we make a balanced binary search tree $T$ whose leaves
	are $1,\dots, n$, from left to right.
	For each node $\nu$ of $T$, we define $C(\nu)$ as
	the set of indices stored at the leaves of the subtree 
	rooted at $\nu$. The set $C(\nu)$ is a \emph{canonical
	subset} of $\{ 1,\dots, n\}$.
	
	For each node $\nu \in T$, we compute the region
	$R(\nu)=\bigcup_{i\in C(\nu)} U_i$.
	This can be done in $O(n \log n)$ time \emph{for all} 
	nodes $\nu$ of $T$. Indeed, the divide-and-conquer
	approach from the proof of Lemma~\ref{le:union_pierced}
	can be applied here.
	If a node $\nu$ has children $\nu_\ell$ and $\nu_r$,
	then $R(\nu)=\bigcup_{i\in C(\nu)} U_i$ can be computed
	in $O(|C(\nu)|)$ time
	from $R(\nu_\ell)=\bigcup_{i\in C(\nu_\ell)} U_i$ and 
	$R(\nu_r)=\bigcup_{i\in C(\nu_r)} U_i$.
	
	For each node $\nu$ of $T$,
	we preprocess the region $R(\nu)$
	for point location queries.
	This takes $O(|C(\nu)|)$ time, because
	we just need the description of the boundary of $R(\nu)$
	in a table.
	To decide whether a given point $x \in \RR^2$ lies in  
	$R(\nu)$, we make a binary search along
	the boundary of $R(\nu)$ for the arc of $R(\nu)$ that
	is intersected by the ray from $q$ through $x$. 
	This takes $O(\log n)$ time.
	This finalizes the preprocessing and the construction of
	the data structure.

	Consider a query consisting of a point $x\in \RR^2$
	and an index $a$. We may assume that $a \in \{ 1,\dots,n \}$.
	The set $\{a,a+1,\dots,n\}$ can be expressed as
	the disjoint union of canonical subsets $C(\nu_1),\dots, C(\nu_k)$,
	where $k=O(\log n)$, indexed so that
	each element of $\nu_t$ is smaller than each element of $\nu_{t+1}$,
	for $t = 1, \dots, k - 1$.
	Making point location queries in $R(\nu_1),R(\nu_2),\dots$
	we find the first index $j$ such that $x\in R(\nu_j)$. This
	takes $O(\log^2 n)$, as we make $O(\log n)$ point location queries.
	
	Then, we search the subtree of $T$
	rooted at $\nu_j$ for the leftmost leaf $i$ with $x\in \UU_i$.
	This is easy: if we are at some internal
	node $\nu$ with left child $\nu_\ell$ and right child $\nu_r$,
	we query the point location data structure at $\nu_\ell$ to 
	determine whether $x\in R(\nu_\ell)$. If $x\in R(\nu_\ell)$,
	we continue to $\nu_\ell$. Otherwise,
	$x$ must be in $R(\nu_r)$, as $x\in R(\nu)$,
	and we go to $\nu_r$.
	This search makes $O(\log n)$ queries
	to the point location structures, and thus takes $O(\log^2 n)$ time.	
\end{proof}

\begin{lemma}
\label{le:data_structure_translates_2}
	Let $\UU_q=\{ V_1,\dots V_n\}$ be a family of $n$ translates 
	of a convex object in the plane 
	that are pierced by a given point $q$. 
	Let $U_0$ be a convex object.
	In $O(n\log^2 n)$ time, we can construct a data structure
	for the following type of queries:
	given a translate $U$ of $U_0$ and a value $a$, 
	find the smallest $i\ge a$ such
	that $U$ intersects $V_i$, or correctly report
	that $U$ does not intersect $V_a\cup\dots\cup V_n$.
	Each query can be answered in $O(\log^2 n)$ time. 
\end{lemma}
\begin{proof}
	Applying a translation, we may assume that $U_0$ 
	contains the origin.
	For each $V_i\in \UU_q$, let 
	$W_i=V_i\oplus U_0 = \{ v-u\mid v\in V_i,\, u\in U_0\}$
	be the Minkowski sum of $V_i$ and $-U_0$.
	For each translation $\tau$,
	we have that  $\tau(U_0)$ intersects $V_i$ if and only if 
	$\tau\in W_i$.
	
	All the sets $W_1,\dots,W_n$ contain $q$
	because the origin belongs to $U_0$.
	Thus, we can construct the data structure of 
	Lemma~\ref{le:data_structure_translates_1} for
	$\{W_1,\dots,W_n\}$. 
	For a query $U$ and $a$, we find the translation $\tau$
	such that $U=\tau(U_0)$, and then
	find the smallest $i\ge a$ such that $\tau\in W_i$,
	which also tells the smallest $i\ge a$ such
	that $U$ intersects $V_i$.
\end{proof}

Lemma~\ref{le:data_structure_translates_2} can be used
to make queries and simulate deletions.

\begin{proposition}
\label{prop:sparsification_translates}
	Consider a family $\UU$ of $n$ 
	translates of a convex object with non-empty interior
	in the plane.
	In $O(n\log^2 n)$ time,
	we can reduce the problem of finding
	a maximum matching in $G_\UU$ to the problem of finding
	a maximum matching in $G_\WW$ for some subfamily
	$\WW\subseteq \UU$ with maximum depth $O(1)$.
\end{proposition}
\begin{proof}
	As mentioned above, we may make an affine transformation, 
	so that, after the transformation, we have 
	$\Psi=O(1)$~\cite[Lemma 3.2]{ahv-aemp-04}.
	Consider an edge $pq$ of the pattern graph.
	We use the algorithm described in Lemma~\ref{le:data_structure},
	but with a slight modification.
	We order the objects of $\UU_q$ as $\{ V_1,\dots,V_m\}$
	and use the data structure of 
	Lemma~\ref{le:data_structure_translates_2} to store them.
	At the start we set $a=1$.
	Whenever we want to query $\AA_q$ with $U$, we query
	the data structure with $U$ and the current $a$.
	If the data structure returns $V_i$, 
	we set $a=i+1$ for future queries
	to the data structure. In this way, each time we query the
	data structure, we find a new element of $\UU_q$ that
	has not been reported before.
	Thus, we obtain the same running time as in
	Lemma~\ref{le:data_structure} with 
	$T_{\rm con}(m)=O(m\log m)$,
	$T_{\rm que}(m)=O(\log^2 m)$ and
	$T_{\rm del}(m)=O(1)$.
	Therefore $T_{\rm sparse}(m)=O(m \log^2 m)$,
	and the result follows from 
	Theorem~\ref{thm:sparsification}.
\end{proof}

Combining Proposition~\ref{prop:sparsification_translates}
and Theorem~\ref{thm:main1} we obtain the following.

\begin{theorem}
\label{thm:translates}
	Consider a family $\UU$ of translates of a convex object 
	with non-empty interior in the plane.
	In $O(n^{\omega/2})$ time we can find a matching in $G_\UU$
	that, with high probability, is maximum.
\end{theorem}

If $\UU$ consists of unit disks, the sparsification can be done
slightly faster using a semi-dynamic data structure by 
Efrat, Itai, and Katz~\cite{EfratIK01},
which has 
$T_{\rm con}(m)=O(m\log m)$, 
$T_{\rm del}(m)=O(\log m)$ and $T_{\rm que}(m) = O(\log m)$.
However the current bottleneck is the computation
of the maximum matching \emph{after} the sparsification.
Thus, improving the sparsification in the particular
case of unit disks does not lead to an improved final algorithm.

Proposition~\ref{prop:sparsification_translates}
and Theorem~\ref{thm:translates}
also holds if we have translations of $O(1)$ different 
convex objects (with nonempty interiors). Indeed, the data structure of 
Lemma~\ref{le:data_structure_translates_2} can be
made for each pair of different convex shapes. 
In this case, the constant $\Psi$ depends on the shapes,
namely the size of the largest square that
we can place inside each of the convex shapes 
and the size of the smallest square that can be used
to cover each of the convex shapes. Also, the relation
between the depth and the density depends on the shapes. 
However, for a fixed set of $O(1)$ shapes, 
both values are constants that depend
on the shapes.

\begin{theorem}
\label{thm:translates2}
        Suppose we are given a set $\mathcal{A}$ of $O(1)$ different convex
	objects in the plane with non-empty interiors. 
	Let 
	$\UU$  be a family that contains $n$ translates of objects from 
	$\mathcal{A}$.
	Then, we can find in $O(n^{\omega/2})$ time a matching in $G_\UU$
	that is maximum with high probability. Here, 
	the constant in the $O$-notation
	depends on $\mathcal{A}$.
\end{theorem}

\subsection{Axis-parallel objects}
A \DEF{box} is the Cartesian product of intervals.
Combining standard data structures for orthogonal range 
searching~\cite[Sections 5.4 and 10.3]{BergCKO08}
one obtains the following results.

\begin{proposition}
\label{prop:axis_parallel}
	Let $d\ge 2$ be an integral constant.
	Consider a family $\UU$ of $n$ boxes in $\RR^d$
	such that each box of $\UU$ contains a cube
	of side length $1$ and is contained in a cube 
	of side length $\Psi$.
	In $O(\Psi^d \cdot n \log^{O(d)} n)$ time we can reduce 
	the problem of finding
	a maximum matching in $G_\UU$ to the problem of finding
	a maximum matching in $G_\WW$, for some $\WW\subseteq \UU$
	with maximum depth $(1+\Psi)^{O(d)}$.
\end{proposition}
\begin{proof}
Edelsbrunner and Maurer~\cite{EdelsbrunnerM81}
show a general approach to provide a data structure
to dynamically maintain a set of boxes and handle
the following queries: given a box $b$, 
report all the boxes in the data structure 
that intersect $b$.
The construction time is $O(n \log^d n)$, each
update (deletion/insertion) takes $O(\log^d)$ time,
and each query takes $O(k+\log^d n)$,
where $k$ is the size of the output.
The data structure is a combination of segment 
and range trees. Such a data structure can easily be
modified to report a single element intersecting
the query box $b$ in $O(\log^d n)$ time. In fact,
better results can be obtained with more advanced
techniques, but we feel that discussing them is not
relevant here. (Also, we only need deletions,
which makes it simpler, as in the relevant trees 
we can just mark some vertices as deleted.)
Using Lemma~\ref{le:data_structure} and
Theorem~\ref{thm:sparsification_d}, we obtain the result.
\end{proof}

For $d=2$, we can combine 
Theorem~\ref{thm:main1} and 
Proposition~\ref{prop:axis_parallel}.
Since we have assumed $\omega>2$,
the $O(n\log^{O(d)} n)$ term is asymptotically smaller
than $O(n^{\omega/2})$, and we obtain the following.

\begin{theorem}
\label{thm:axis-parallel}
	Given a family $\UU$ of $n$ boxes in $\RR^2$
	such that each object of $\UU$ contains a square
	of side length $1$ and is contained in a square 
	of side length $\Psi$, we can compute 
	in $(1+\Psi)^{O(1)} n^{\omega/2}$ time
	a matching in $G_\UU$ that, with high probability,
	is a maximum matching.
\end{theorem}

Consider now the case $d\geq 3$.
The set $\WW$ that we 
obtain from Proposition~\ref{prop:axis_parallel}
has depth and density $\rho=(1+\Psi)^{O(d)}$,
and therefore the graph $G_\WW$ has $O(\rho n)$ edges;
see Lemma~\ref{lem:edges}.
We can thus use the algorithm of Micali and 
Vazirani~\cite{MicaliV80,Vazirani12},
which takes 
$O(\sqrt{n}|E(G_\WW)|) = (1+\Psi)^{O(d)} n^{3/2}$ time.
We summarize.\footnote{In a previous version,
we invoked M\k{a}dry's algorithm and 
claimed a better running time with exponent
$10/7$. However, this algorithm does not seem to
apply here, so we fall back on the Micali-Vazirani
algorithm to obtain an exponent of $3/2$.}

\begin{corollary}
\label{cor:axis-parallel}
	Let $d\ge 3$ be an integral constant. 
	Given a family $\UU$ of $n$ boxes in $\RR^d$
	such that each object of $\UU$ contains a cube
	of side length $1$ and is contained in a cube 
	of side length $\Psi$, we can compute 
	in $(1+\Psi)^{O(d)} n^{3/2}$ time
	a maximum matching in $G_\UU$.
\end{corollary}

\subsection{Congruent balls in \texorpdfstring{$d\ge 3$}{three or more} 
dimensions}
Consider now the case of congruent balls in $\RR^d$, for constant $d\ge 3$.
Note that $\lambda=O(1)$ in this case.
We use the dynamic data structure by Agarwal and 
Matou{\v s}ek~\cite{AgarwalM93} 
for the sparsification. 
For each $m$ with $n\le m\le n^{\lceil  d/2\rceil}$, the data structure
maintains $n$ points in $\RR^d$, answers $O(n)$ queries for closest
point and supports $O(\lambda^2)$ updates in
\[
	O\left(m^{1+\varepsilon}+ \lambda^2 \frac{m^{1+\varepsilon}}{n} + 
			n\cdot \frac{n \log^3 n}{ m^{1/ \lceil d/2\rceil}}\right)
\] 
time. Here $\varepsilon > 0$, is an arbitrary constant
whose choice affects to the constants hidden in the $O$-notation.
For $d \in \{3, 4\}$, this running time is 
\[
	O\left(m^{1+\varepsilon}+ \lambda^2 \frac{m^{1+\varepsilon}}{n} + 
			n\cdot \frac{n \log^3 n}{ m^{1/2}}\right).
\]
Setting $m=n^{4/3}$, we get a running time of
$O(n^{4/3+\varepsilon} + \lambda^2 n^{1/3+\varepsilon})=O(n^{4/3+\varepsilon})$
to handle $O(n)$ queries and $O(\lambda^2)=O(1)$ updates.
Using this in Lemma~\ref{le:data_structure}
and Theorem~\ref{thm:sparsification_d},
we get the following result
 
\begin{proposition}
\label{prop:sparsification_balls}
	Consider a family $\UU$ of $n$ unit balls objects in $\RR^d$,
	for $d \in \{3, 4\}$.
	In $O(n^{4/3+\varepsilon})$ time, we can reduce 
	the problem of finding
	a maximum matching in $G_\UU$ to the problem of finding
	a maximum matching in $G_\WW$ for some $\WW\subseteq \UU$
	with maximum depth $O(1)$.
\end{proposition}

For the resulting set $\WW$ with depth $O(1)$, it is better
to use the algorithm of Micali and Vazirani~\cite{MicaliV80,Vazirani12}.
Note that $G_\WW$ is sparse, and thus has $O(n)$ edges.
Therefore, a maximum matching in $G_\WW$ can be computed in $O(n^{3/2})$ time.
In summary, we spend $O(n^{4/3+\varepsilon})$ for the sparsification
and $O(n^{3/2})$ for computing the matching in the sparsified 
setting.\footnote{Also here, we previously claimed a better 
exponent of $10/7$, which was based on an incorrect
invocation of M\k{a}dry's algorithm.}

For $d > 4$, we set 
$m=n^{\frac{2 \lceil d/2\rceil}{1+ \lceil d/2\rceil}}$.
The running time for the sparsification is then
$O(n^{\frac{2 \lceil d/2\rceil}{1+ \lceil d/2\rceil}+\varepsilon})$.
For each constant $d$, the resulting instance $G_\WW$ has $O(n)$
edges. For $d=5,6$, the running time of the sparsification is 
$O(n^{3/2+\varepsilon})$. However, after the sparsification,
we have a graph with $O(n)$ edges, and we can use the
algorithm of Micali and Vazirani~\cite{MicaliV80,Vazirani12}, 
which takes $O(n^{3/2})$ time.
Thus, for $d\ge 5$, the running time
is dominated by the sparsification.

\begin{theorem}
\label{thm:matching_balls}
	Let $d\ge 3$ be a constant.
	Consider a family $\UU$ of congruent balls in $\RR^d$.
	For $d=3,4$, we can find in $O(n^{3/2})$ time a 
	maximum matching in $G_\UU$.
	For $d\ge 5$, we can find in 
	$O(n^{\frac{2 \lceil d/2\rceil}{1+ \lceil d/2\rceil}+\varepsilon})$
	time a maximum matching in $G_\UU$, for each $\varepsilon>0$.
\end{theorem}

\section{Conclusion}

We have proposed the density of
a geometric intersection graph as a parameter
for the maximum matching problem, and we showed 
that it can be fruitful in obtaining efficient
matching algorithms. Then, we presented a sparsification
method that lets us reduce the general problem to
the case of bounded density for several interesting classes
of geometric intersection graphs.
In our sparsification method, we did not attempt
to optimize the dependency on the radius ratio $\Psi$.
It may well be that this can be improved by using more
advanced grid-based techniques.
Furthermore, 
our sparsification needs the complete intersection graph
and does not apply to the bipartite setting. Here,
we do not know of a method to reduce the general case 
to bounded density.
In general, the complexity of the matching problem
is wide open. To the best of our knowledge, there are 
no (even weak) 
superlinear lower bounds for the (static)
matching problem in general graphs.


\bibliography{bibliodisks}

\end{document}